\documentclass[10pt, journal]{IEEEtran}

\usepackage[utf8]{inputenc}
\usepackage[T1]{fontenc}
\usepackage{lmodern}
\usepackage[hyphens]{url}
\usepackage{doi}

\usepackage{newunicodechar}
\newunicodechar{ρ}{\rho}
\newunicodechar{σ}{\sigma}
\newunicodechar{λ}{\lambda}
\newunicodechar{Λ}{\Lambda}
\newunicodechar{μ}{\mu}
\newunicodechar{ν}{\nu}
\newunicodechar{ψ}{\psi}
\newunicodechar{ϕ}{\phi}
\newunicodechar{φ}{\varphi}
\newunicodechar{π}{\pi}
\newunicodechar{α}{\alpha}
\newunicodechar{ϵ}{\epsilon}
\newunicodechar{ε}{\varepsilon}
\newunicodechar{δ}{\delta}
\newunicodechar{ω}{\omega}
\newunicodechar{∗}{\textasteriskcentered}

\usepackage{microtype}

\usepackage[english]{babel}
\usepackage{csquotes}

\usepackage{xparse}
\usepackage{xspace}
\usepackage{makecell}

\usepackage{versions}
\includeversion{DRAFT}
\excludeversion{TRASH}
\excludeversion{EXCLUDED}

\usepackage{amssymb}
\usepackage{amsmath}
\usepackage{mathrsfs}
\usepackage{mathtools}
\usepackage{optidef}

\usepackage{cool}
\usepackage{xcolor}
\usepackage{hyperref}
\hypersetup{
	colorlinks,
	linkcolor={red!50!black},
	citecolor={red!50!black},
	urlcolor={blue!80!black},
	linktoc=all
}
\usepackage[nameinlink]{cleveref}

\usepackage{amsthm}
\usepackage{thmtools}
\newtheorem{theorem}{Theorem}
\newtheorem{lemma}[theorem]{Lemma}
\newtheorem{corollary}[theorem]{Corollary}

\theoremstyle{definition}
\newtheorem{remark}[theorem]{Remark}

\newtheorem*{definition*}{Definition}

\usepackage{graphicx}         
\setkeys{Gin}{width=\textwidth} %
\usepackage[multidot,spae]{grffile} 
\graphicspath{{./}{img/}}

\usepackage{suffix}
\usepackage{xparse}
\usepackage{braket}

\let\bkbraket\braket

\usepackage{physics}

\renewcommand{\tr}{\Tr}

\usepackage{mathtools}
\usepackage{amssymb}
\usepackage{tensor}
\usepackage[binary-units=true]{siunitx}

\usepackage{ stmaryrd }

\newcommand{\br}[1]{\left(#1\right)}

\NewDocumentCommand{\stared}{m m}{\IfBooleanTF{#2}{#1*}{#1}}
\NewDocumentCommand{\es}{m o}{
	\IfNoValueTF{#2}{
		\mathbb{\uppercase{#1}}^{\lowercase{#1\/}}
	}{
		\mathbb{\uppercase{#1}}^{#2\/}
	}
}

\newcommand{\complex}{\mathbb{C}}
\newcommand{\reals}{\mathbb{R}}

\newcommand{\naturals}{\mathbb{N}}

\newcommand{\HS}{\mathcal{H}}

\NewDocumentCommand{\BO}{o}{\mathcal{B}\br{\IfValueTF{#1}{#1}{\HS}}}

\NewDocumentCommand{\DM}{o}{\mathcal{D}\br{\IfValueTF{#1}{#1}{\HS}}}

\NewDocumentCommand{\pos}{o}{\mathscr{P}\!\br{\IfValueTF{#1}{#1}{\HS}}}

\NewDocumentCommand{\supp}{m}{\textrm{supp}\br{#1}}

\newcommand{\poly}{\mathrm{poly}}

\DeclarePairedDelimiterX{\infdivx}[2]{(}{)}{
	#1\;\delimsize|\delimsize|\;#2
}

\DeclareDocumentCommand{\qrd}{o o m m}{\ensuremath{\IfValueTF{#1}{#1}{D}\IfValueTF{#2}{_{#2}}{}\/\infdivx*{#3}{#4}}}

\NewDocumentCommand{\ent}{d()g}{\ensuremath{H\br{\IfValueTF{#1}{#1}{#2}}}}
\NewDocumentCommand{\qent}{d()g}{\ensuremath{H\br{\IfValueTF{#1}{#1}{#2}}}}

\newcommand*\ric[1]{\vphantom{#1}\smash{#1_{}\kern-\scriptspace}}

\newcommand{\foralls}{\forall\,}

\renewcommand{\abs}[1]{\left\lvert#1\right\rvert}

\DeclarePairedDelimiter{\mabs}{\lvert}{\lvert}

\newcommand{\renyi}{R\'enyi\xspace}
\DeclareMathOperator*{\argmin}{\arg\min}
\DeclareMathOperator*{\argmax}{\arg\max}

\newcommand{\Sn}{\mathfrak{S}_n}
\NewDocumentCommand{\End}{o}{\mathrm{End}^{\Sn}(\IfValueTF{#1}{#1}{R^n B^n})}

\newcommand{\n}{^{\otimes n}}

\title{Parallelization of Adaptive Quantum Channel Discrimination in the Non-Asymptotic Regime}

\renewcommand{\E}{\mathcal{E}}
\newcommand{\F}{\mathcal{F}}
\newcommand{\id}{\textrm{id}}

\author{Bjarne Bergh\thanks{Bjarne Bergh, Nilanjana Datta and Robert Salzmann are with the Department of Applied Mathematics and Theoretical Physics, University of Cambridge, United Kingdom (e-mails: bb536@cam.ac.uk, n.datta@statslab.cam.ac.uk, rals.salzmann@web.de).}, Nilanjana Datta, Robert Salzmann, Mark M. Wilde\thanks{Mark M. Wilde is with the Hearne Institute for Theoretical Physics, Department of Physics and Astronomy, Center for Computation and Technology, Louisiana State University, Baton Rouge, Louisiana 70803, USA, and the School of Electrical and Computer Engineering, Cornell University, Ithaca, New York 14850, USA (e-mail: wilde@cornell.edu).}}

\begin{document}
	\maketitle
	\begin{abstract}
		We investigate the performance of parallel and adaptive quantum channel discrimination strategies for a finite number of channel uses. It has recently been shown that, in the asymmetric setting with asymptotically vanishing type~I error probability, adaptive strategies are asymptotically not more powerful than parallel ones. We extend this result to the non-asymptotic regime with finitely many channel uses, by explicitly constructing a parallel strategy for any given adaptive strategy, and bounding the difference in their performances, measured in terms of the decay rate of the type~II error probability per channel use. We further show that all parallel strategies can be optimized over in time polynomial in the number of channel uses, and hence our result can also be used to obtain a poly-time-computable asymptotically tight upper bound on the performance of general adaptive strategies.
	\end{abstract}
	\section{Introduction}
	
	Binary quantum state discrimination is one of the oldest and most studied tasks in quantum information theory \cite{helstrom_quantum_1969,holevo_analogue_1972}. The task involves determining the state of a quantum system, given the side information that it is in one of two possible states. This is done by performing suitable measurements on the state. This task has been studied in both the $n$-shot regime, in which a finite number ($n$) of identical copies of the unknown state are available, and in the asymptotic limit in which one assumes that an infinite number of copies are available (i.e., $n \to \infty$). The task of binary state discrimination, which can be viewed as that of binary hypothesis testing for quantum states, is a fundamental primitive of quantum information theory since many other information processing tasks can be reduced to it. Known results for this problem now include optimal decision strategies and the behavior of the error probabilities of misidentification in both these finite and asymptotic regimes~\cite{helstrom_quantum_1969, kholevo_asymptotically_1979, hiai_proper_1991, ogawa_strong_2000, nagaoka_converse_2006, audenaert_discriminating_2007, hayashi_error_2007, audenaert_asymptotic_2008, nussbaum_chernoff_2009, audenaert_quantum_2012, bae_quantum_2015}. There are generally two possible errors that can be incurred: the system is inferred to be in the second state when it is actually in the first ({\em{type~I error}}) or vice versa ({\em{type~II error}}). There are two different settings in which the state discrimination task is usually studied: in the {\em{symmetric setting}}, the probabilities of these two errors are treated on an equal footing, while in the {\em{asymmetric setting}}, one minimizes the probability of type~II error under the constraint that the type~I error is below a given threshold.

	Similar in nature, but a lot less understood, is the task of binary quantum channel discrimination: given an unknown quantum channel as a black box and the side information that it is one of two possible channels, the task is to determine the channel's identity \cite{chiribella_memory_2008,duan_perfect_2009,hayashi_discrimination_2009,harrow_adaptive_2010}.
	The additional complexity here comes from the fact that, on top of finding the best measurement to perform on the output of the channel, we also have to figure out which quantum states to send as input to the channel. Say we are given access to $n$ copies of the channel (i.e., we are given $n$ identical black boxes, each of which can be used once); then there are different strategies (sometimes also called protocols) in which we could set up our decision experiment -- the so-called {\em{parallel}} and {\em{adaptive}} strategies.
	In a parallel strategy one prepares a joint state, usually entangled between the input systems of all the $n$ copies of the channel and an additional reference (or memory) system. This state is then fed as input to all the $n$ channels at once (with the state of the reference system being left undisturbed). Finally, a binary positive operator-valued measure (POVM) is performed on the joint state at the output of the channels and the reference system in order to arrive at a decision for the channel's identity.
	In an adaptive strategy, on the other hand, one prepares a state of the input system of a single copy of the channel (again usually entangled with a reference system) which is fed into the first copy of the channel, with the state of the reference system being left undisturbed. The input to the next use of the channel is then chosen depending on the output of the first channel and the state of our reference system. This is done, most generally, by subjecting the latter to an arbitrary quantum operation (or channel), which we call a preparation operation. This step is repeated for each successive use of the channel until all the $n$ black boxes have been used.
	Then a binary POVM is performed on the joint state of the output of the last use of the channel and the reference system. See \autoref{fig:adaptive_strategy} for a depiction of an adaptive strategy. Adaptive strategies are also sometimes called sequential, which  however should not be confused with the setting of sequential hypothesis testing \cite{martinez_vargas_quantum_2021, li_optimal_2022, li_sequential_2022}, where samples (i.e., states or channels) can be requested one by one. 
	
	One particularly interesting question is whether and to what degree adaptive strategies give an advantage over parallel ones. Note that any parallel strategy can be written as an adaptive strategy by taking all but one channel input as part of the reference system, and then choosing each preparation operation such that it extracts the next part of the joint input state for the next channel use and replaces it by the output of the previous channel use. However, the converse is not true, and so adaptive strategies are more general. Parallel strategies are conceptually a lot simpler than adaptive ones -- aside from the measurement, everything is specified just by the joint input state -- in contrast to adaptive strategies, in which after each channel use we can perform an arbitrary quantum operation to prepare the input to the next use of the channel. It is thus interesting to determine to what degree parallel strategies can still be optimal. It is known that in certain cases adaptive strategies can give an advantage over parallel ones. In~\cite{harrow_adaptive_2010} the authors constructed an example in which an adaptive strategy with only two channel uses could be used to discriminate the channels with certainty, which is not possible with a parallel strategy, even if arbitrarily many channel uses are allowed.
	
	Interestingly, \emph{asymptotically}, there are multiple known cases in which adaptive strategies give no advantage over parallel ones, i.e., the optimal exponential decay rate of the error probability per channel use is the same in the asymptotic limit. For example, this is the case both in the symmetric and asymmetric settings when the channels are classical \cite{hayashi_discrimination_2009} or classical-quantum \cite{wilde_amortized_2020,salek_usefulness_2022}. For arbitrary quantum channels, the recently shown chain rule for the quantum relative entropy \cite{fang_chain_2020} and the characterization of asymmetric channel discrimination in terms of amortized relative entropy \cite{wilde_amortized_2020,wang_resource_2019}, also imply that in the asymmetric setting, in the asymptotic limit where we also require the type~I error to vanish, adaptive strategies give no advantage over parallel ones (i.e., the optimal asymptotic exponential decay rate of the type~II error per channel use is the same for parallel and adaptive strategies; see \Cref{sec:previous_things} below for more details). This is in contrast to the symmetric setting in which the example of \cite{salek_usefulness_2022} shows that there always is an advantage of adaptive strategies, also in the asymptotic limit. 
	
	In this paper, our aim is to move from these purely asymptotic results to statements comparing adaptive and parallel strategies for finite $n$.  In our main result (\autoref{cor:seq_par} and \autoref{thm:seq_par}), we relate the type~I and type~II errors of an arbitrary adaptive strategy with those of a suitably chosen parallel strategy. Specifically, given an adaptive strategy with $n$ channel uses, we construct a parallel strategy with $m$ channel uses (where $m$ can be chosen at will), such that for arbitrary fixed type~I errors of the two strategies, we find an explicit bound on the difference between their type~II error decay rates. 
	This difference goes to zero in a suitably chosen asymptotic limit $m, n \to \infty$ if also the type~I error vanishes, hence also implying the known asymptotic equivalence in the asymmetric case. Our result answers the following interesting question which is of practical relevance: {\em{Given an adaptive strategy involving $n$~uses of the channel, if one instead employs a parallel strategy with $m$ channel uses, how much worse are the errors going to be?}}
	
	Note that the asymptotic results obtained in~\cite{wilde_amortized_2020, wang_resource_2019} do not purely come from finite-length considerations, and hence results analogous to ours are not directly obtainable from these references. 
	
	Our main result becomes particularly interesting considering that we additionally show that one can optimize over all parallel strategies in time polynomial in the number of channel uses (\autoref{lem:parallel_polytime}), whereas optimizing over all adaptive strategies seems to require exponential time (see \autoref{rem:computability} below). Hence we can also use our theorem to give a bound on the following practically relevant question: {\em{Given that one computed the optimal parallel strategy involving $m$ uses of the channel, how much better could the errors of the best adaptive strategy with $n$ channel uses possibly be?}}
	
	Note that the upper bound on the finite-$n$ performance of adaptive strategies that can be obtained like this by optimizing over all parallel strategies, is to our knowledge the first upper bound on this quantity that is computable in time polynomial in the number of channel uses and also asymptotically tight. 
	
	On the way to proving our main theorem, we also prove the following two technical results: First, a slightly refined bound on the difference between the Petz--\renyi relative entropy $D_{\alpha}$ and the Umegaki relative entropy $D$ (\autoref{lem:renyi_alpha_continuity}), and secondly, a bound on the convergence speed in the asymptotic equipartition property of the smoothed max-relative entropy (\autoref{lem:aep_finite_n}), which turns out to be better suited to our application than previous results (see \autoref{remark:relation_to_tomamichel} and \autoref{remark:relation_to_second_order} for a detailed explanation of the differences).
	
	We start by introducing all of the required notation and previous results in \Cref{sec:previous_things}. Then we state and prove our main theorem (\autoref{thm:seq_par}) in \Cref{sec:main_result}. \Cref{sec:example} illustrates how adaptive and parallel strategies can differ in the finite regime with an example, and it also demonstrates an application of our result. 
	
	\section{Preliminaries, Notation, and Previous Results} \label{sec:previous_things}
	\subsection{Notation}
	
	We write $\HS$ for a complex finite-dimensional Hilbert space and $\BO$ for the set of linear operators acting on $\HS$. We write $\pos$ for the set of positive semi-definite operators acting on $\HS$. For $A, B \in \pos$ we further write $A \ll B$ if $\supp{A} \subseteq \supp{B}$  and $A \not \ll B$ if $\supp{A} \not \subseteq \supp{B}$.  Let $\DM$ denote the set of density matrices, i.e., the set of positive semi-definite operators with trace one. A quantum channel (in this paper usually denoted as $\E$ or $\F$) is a completely positive, trace-preserving map between density operators. We will label different quantum systems by capital Roman letters ($A$, $B$, $C$, etc.) and often use these letters interchangeably with the corresponding Hilbert space or set of operators or density matrices (i.e., we write $\rho \in \DM[A]$ instead of $\rho \in \DM[\HS_A]$ and $\E: A \to B$ instead of $\E: \DM[\HS_A] \to \DM[\HS_B]$).  We write $\log$ for the logarithm to the base two.
	
	\subsection{Quantum Information Measures}
	
	For $\rho \in \DM$ and $\sigma \in \pos$ the (Umegaki) quantum relative entropy is defined as \cite{umegaki_conditional_1962}
	\begin{equation}
		D(\rho\|\sigma) \coloneqq \Tr(\rho(\log \rho - \log \sigma)),
	\end{equation}
	if $\rho \ll \sigma$ and 
	$D(\rho\|\sigma) \coloneqq \infty$ if $\rho \not \ll \sigma$. One of its most important properties is the data-processing inequality \cite{lindblad_completely_1975}, which states that for every quantum channel $\E$:
	\begin{equation}
		D(\rho\|\sigma) \geq D(\E(\rho)\|\E(\sigma)) \,.
		\label{eq:data-proc-ineq}
	\end{equation}
	A self-contained proof can be found, e.g., in \cite{khatri_principles_2020}.
	More generally, we call a function of $\rho$ and $\sigma$ a divergence if it satisfies the data-processing inequality.
	
	Below we also make use of the binary entropy, which is the Shannon entropy of a classical random variable that takes two possible values. It is uniquely specified by the probability $p \in [0,1]$ of one of the values and is given by
	\begin{equation}
		h(p) \coloneqq - p \log p - (1-p) \log (1-p)\,.
	\end{equation}
	It satisfies
	\begin{equation}
		0 \leq h(p) \leq 1 \qquad \foralls p \in [0,1]\,.
	\end{equation}
	
	\subsubsection{Fidelity and Sine Distance}
	
	For two quantum states $\rho, \sigma \in \DM$, we define the fidelity as \cite{uhlmann_transition_1976} 
	\begin{equation}
		F(\rho, \sigma) \coloneqq \Tr(\sqrt{\sqrt{\sigma}\rho\sqrt{\sigma}}).
	\end{equation}
	Note that this is sometimes also called the square root fidelity due to different conventions on whether to include a square or not. 
	We define the sine distance as \cite{rastegin_relative_2002,rastegin_lower_2003,gilchrist_distance_2005,rastegin_sine_2006}
	\begin{equation}
		P(\rho, \sigma) \coloneqq \sqrt{1 - F(\rho, \sigma)^2},
	\end{equation}
	which satisfies the properties of a metric (especially the triangle inequality) and also the data-processing inequality (see, e.g., \cite{khatri_principles_2020}). This quantity is also known as the purified distance.
	
	\subsubsection{(Smoothed) Max-Divergence}
	
	For $\rho \in \DM$ and $\sigma \in \pos$, define the quantum max-divergence (or the max-relative entropy) as \cite{datta_min-_2009}
	\begin{equation}
		D_{\max}(\rho\|\sigma) \coloneqq \log \inf\Set{\lambda \in \reals | \rho \leq \lambda \sigma}\,.
	\end{equation}
	The quantum max-divergence also satisfies the data-processing inequality \cite{datta_min-_2009}.
	Let
	\begin{equation}
		B^{\circ}_{\varepsilon}(\rho) \coloneqq \Set{\tilde{\rho} \in \DM | P(\rho, \tilde{\rho}) \leq \varepsilon}
	\end{equation}
	be the $\varepsilon$-ball of \emph{normalized} states around $\rho$ in sine distance. 
	Then, we define the smoothed max-divergence as  \cite{datta_min-_2009}
	\begin{equation}
		D_{\max}^{\varepsilon}(\rho\|\sigma) \coloneqq D_{\max}^{\varepsilon, \circ}(\rho\|\sigma) = \inf_{\tilde{\rho} \in B^{\circ}_{\varepsilon}(\rho)} D_{\max}(\tilde{\rho}\|\sigma).
	\end{equation}
	Note that this is sometimes defined differently in the literature, where one allows the infimum to also include sub-normalized states. 
	
	\subsubsection{\renyi Divergences}
	
	Let $\rho \in \DM$ and $\sigma \in \pos$ be two operators such that $\rho \ll \sigma$. The definitions below can, in some cases, also be extended with finite values to the case where $\rho \not\ll \sigma$ (see the references below or also \cite{khatri_principles_2020}); however, this is not going to be relevant for our work.
	For every $\alpha \in [0,1) \cup (1,2]$ define the \emph{Petz--\renyi divergence} as \cite{petz_quasi-entropies_1986}:
	\begin{equation}
		D_\alpha(\rho\|\sigma) \coloneqq \frac{1}{\alpha - 1} \log \Tr(\rho^{\alpha}\sigma^{1-\alpha})\,.
	\end{equation}
	Similarly, for every $\alpha \in (0,1) \cup (1,2]$ define the \emph{geometric \renyi divergence} as \cite{matsumoto_new_2018}:
	\begin{equation}
		\widehat{D}_\alpha(\rho\|\sigma) \coloneqq \frac{1}{\alpha - 1} \log \Tr\!\left(\sigma^{1 \over 2}( \sigma^{-{1 \over 2}} \rho \sigma^{-{1 \over 2}})^\alpha \sigma^{1 \over 2}\right)\,.
	\end{equation}
	It was shown in \cite{matsumoto_new_2018} that the geometric \renyi divergence is the largest \renyi divergence for every $\alpha \in (0,1) \cup (1,2]$, and so
	\begin{equation}
		D_\alpha(\rho\|\sigma) \leq \widehat{D}_\alpha(\rho\|\sigma)\,.
		\label{eq:petz-to-geometric}
	\end{equation}
	
	\subsubsection{Channel Divergences}
	
	We say that a function $\mathbf{D}$ of $\rho \in \DM$ and $\sigma \in \pos$ is a divergence if it satisfies the data-processing inequality (recall \eqref{eq:data-proc-ineq} here).
	For every given divergence $\mathbf{D}$ for states, one can define an associated channel divergence \cite{leditzky_approaches_2018} by performing a (stabilized) maximization over all input states, i.e., with $\E, \F: A \to B$ being quantum channels
	\begin{equation}
		\mathbf{D}(\E\|\F) \coloneqq \sup_{\rho_{RA} \in \DM[R \otimes A]}\mathbf{D}\big((\id_R \otimes \E)(\rho)\|(\id_R \otimes \F)(\rho)\big).
	\end{equation}
	Since $\mathbf{D}$ satisfies the data-processing inequality by definition, the supremum can be restricted to pure states such that the reference system $R$ is isomorphic to the channel input system $A$. 
	
	Specifically relevant later will be the max channel divergence, which has the following simple representation (see \cite[Definition~19]{diaz_using_2018} and \cite[Lemma 12]{wilde_amortized_2020}):
	\begin{align}
		D_{\max}(\E\|\F) &\coloneqq \sup_{\psi_{RA}} D_{\max}\big((\id_R \otimes \E)(\psi)\|(\id_R \otimes \F)(\psi)\big)\\
		&= D_{\max}\big((\id_R \otimes \E)(\Phi)\|(\id_R \otimes \F)(\Phi)\big),
	\end{align}
	where $R \simeq A$ and $\Phi = \Phi_{RA}$ is a maximally entangled state. Hence, the max channel divergence is easily computable as a semidefinite program (SDP).

	\subsection{Quantum State Discrimination}
	
	In binary quantum state discrimination, a quantum system with Hilbert space $\mathcal{H}$ is prepared in one of  two states $\rho$ or $\sigma$, and the objective is to perform a binary POVM $\{\Pi, \IdentityMatrix - \Pi\}$ (where $\IdentityMatrix$ denotes the identity operator acting on $\mathcal{H}$ and $\Pi \in {\mathcal{B}}({\mathcal{H}})$ satisfying $0 \le \Pi \le \IdentityMatrix$) in order to guess which state was prepared (see \cite{bae_quantum_2015} for a review). Performing this measurement comes with the possibility of making an error of misidentification. Generally one defines the type~I and type~II error probabilities, respectively, as
	\begin{align}
		\alpha(\Pi, \rho, \sigma) & \coloneqq \Tr((\IdentityMatrix - \Pi)\rho), \\
		\beta(\Pi, \rho, \sigma) & \coloneqq \Tr(\Pi \sigma)\,.
	\end{align} 
	In the rest of the paper, we refer to the above simply as the type~I and type~II errors.
	It is easy to see from the definition that different choices of $\Pi$ lead to different trade-offs between the type~I and type~II errors. 
	
	In the asymmetric setting, we are interested in finding the optimal type~II error while keeping the type~I error below a threshold (usually denoted as $\varepsilon$). This optimal type~II error is then defined for all $\varepsilon \in [0,1]$ as	 
	\begin{equation}
		\beta_{\varepsilon}(\rho \Vert \sigma) \coloneqq \min_{\substack{0 \leq \Pi \leq \IdentityMatrix\\ \tr(\Pi \rho) \geq 1 - \varepsilon}} \tr(\Pi \sigma).
	\end{equation}
	
	The negative logarithm of this optimal type~II error is called the hypothesis testing relative entropy \cite{buscemi_quantum_2010,brandao_one-shot_2011,wang_one-shot_2012}:
	\begin{equation}
		D_H^{\varepsilon}(\rho\|\sigma) \coloneqq - \log \beta_{\varepsilon}(\rho \Vert \sigma).
	\end{equation}
	It is also known as the smooth min-relative entropy \cite{anshu_building_2019,wang_resource_states_2019}.
	The hypothesis testing relative entropy satisfies the data-processing inequality \cite{wang_one-shot_2012} and can be related to other entropic quantities. 
	We will make use of the following known relations:
	
	\begin{lemma}[Upper bound on $D_H^\varepsilon$]\label{lem:wang_renner_upper_bound}
		Let $\rho, \sigma \in \DM$ be quantum states. Then for all $\varepsilon \in [0, 1)$ 
		\begin{equation}
			D_H^{\varepsilon}(\rho\|\sigma) \leq \frac{1}{1 - \varepsilon}\big(D(\rho\|\sigma) + h(\varepsilon)\big),
		\end{equation}
		where $h(\varepsilon)$ is the binary entropy function.
	\end{lemma}
	This first Lemma is a very well-known consequence of the data-processing inequality of the relative entropy and has been used in converse proofs all throughout information theory. A statement and proof using our notation can be found in \cite{wang_one-shot_2012}, although the essence of the statement can already be found much earlier, for example in \cite[Theorem 2.2]{hiai_proper_1991} and also \cite[Eq.~(3.30)]{hayashi_quantum_2006}. We will also make use of the following Lemma:
	
	\begin{lemma}[Relation to smoothed max-divergence {\cite[Theorem 4]{anshu_minimax_2019}}]\label{lem:anshu_min_max}
		Let $\rho \in \DM$ and $\sigma \in \pos$. Then, for $\varepsilon \in (0,1)$ and $\delta \in (0, 1 - \varepsilon^2)$:
		\begin{multline}\label{eq:anshu_min_max}
			D_H^{1 - \varepsilon ^2 - \delta}(\rho\|\sigma) - \log(4 (1-\varepsilon^2) \over \delta^2) \leq D_{\max}^\varepsilon (\rho\|\sigma) \\ \leq D_{H}^{1 - \varepsilon^2} (\rho\|\sigma) - \log(1 - \varepsilon^2).
		\end{multline}
	\end{lemma}
	
	\begin{remark}
		Note that \cite{anshu_minimax_2019} defines the smoothed max-divergence using sub-normalized smoothing; however, the proof of this lemma goes through also when restricting to normalized states. Note also that the arXiv version (v1) of \cite{anshu_minimax_2019} has a typo in the admissible range of $\delta$, which has been corrected in the published version.
	\end{remark}
	
	\begin{EXCLUDED}
		
		Note that the last relation can be slightly tightened if we can assume $n$ to be large enough:
		
		\begin{lemma}[Better for large/small $\varepsilon$, but only for $n$ large enough]
			If 
			\begin{equation}
				\frac{\log(1-\varepsilon)}{n} \leq \qty(\frac{\log3}{2})^2  \quad \text{and} \quad
				\frac{\log(\varepsilon)}{n} \leq \qty(\frac{\log3}{2})^2
			\end{equation}
			then with $\Upsilon = \Upsilon(\rho\|\sigma)$
			\begin{align}
				\frac{1}{n} D^{\varepsilon}_H(\rho^{\otimes n}\|\sigma^{\otimes n}) &\leq D(\rho\|\sigma) + 4 \log(\Upsilon) \sqrt{\frac{- \log(1 - \varepsilon)}{n}} - \frac{\log(1 - \varepsilon)}{n} \\
				\frac{1}{n} D^{\varepsilon}_H(\rho^{\otimes n}\|\sigma^{\otimes n}) & \geq D(\rho\|\sigma) - 4 \log(\Upsilon) \sqrt{\frac{- \log(\varepsilon)}{n}} + \frac{\log(\varepsilon)}{n}
			\end{align}
		\end{lemma}
		\begin{proof}
			From {\cite[Prop. 4.68 and Prop 4.69]{khatri_principles_2020}} we have the following relations between hypothesis testing and \renyi entropies:
			\begin{align}
				D^\varepsilon_H(\rho\|\sigma) & \leq \bar{D}_\alpha(\rho\|\sigma) + \frac{\alpha}{\alpha - 1}\log(\frac{1}{1 - \varepsilon}) &\qquad  \foralls \alpha &\in (1, \infty) \\
				D^\varepsilon_H(\rho\|\sigma) &\geq \bar{D}_\alpha(\rho\|\sigma) + \frac{\alpha}{\alpha - 1}\log(\frac{1}{\varepsilon})	&\qquad  \foralls \alpha &\in (0,1)
			\end{align}
			Also, from \cite[Lemma 2.3]{audenaert_quantum_2012} (see also \cite[Lemma 6.3]{tomamichel_framework_2012}) we have for $\delta \in (0, {\log 3 \over 4 \log \Upsilon(\rho\|\sigma)})$:
			\begin{align}
				\bar{D}_{1 + \delta}(\rho\|\sigma) &\leq D(\rho\|\sigma) + 4 \delta [\log\Upsilon(\rho\|\sigma)]^2 \\
				\bar{D}_{1 - \delta}(\rho\|\sigma) &\geq D(\rho\|\sigma) - 4 \delta [\log\Upsilon(\rho\|\sigma)]^2
			\end{align}
			In the following we prove the upper bound of our lemma, the lower bound works analogously. We choose $\alpha = 1 + \delta = 1 + \frac{1}{2 \mu \sqrt{n}}$, where
			\begin{equation}
				\mu = \log(\Upsilon) (- \log(1 - \varepsilon))^{-1/2}
			\end{equation}
			
			Combining the above relations we find
			\begin{align}
				\frac{1}{n} D^{\varepsilon}_H(\rho^{\otimes n}\|\sigma^{\otimes n}) &\leq \frac{1}{n} \bar{D}_{1 + \delta}(\rho^{\otimes n}\|\sigma^{\otimes n}) - \frac{1}{n} (1 + \frac{1}{\delta}) \log( 1- \varepsilon) \\
				&= \bar{D}_{1 + \delta}(\rho\|\sigma) - \frac{1}{n} (1 + \frac{1}{\delta}) \log( 1- \varepsilon) \\
				&\leq D(\rho\|\sigma) + 4 \delta (\log \Upsilon)^2 - \frac{1}{n\delta} \log( 1 - \varepsilon)  - \frac{1}{n} \log( 1- \varepsilon) \\
				&= D(\rho\|\sigma) + 4 \log(\Upsilon) \sqrt{\frac{- \log(1 - \varepsilon)}{n}} - \frac{\log(1 - \varepsilon)}{n}
			\end{align}
		\end{proof}
		
	\end{EXCLUDED}
	
	\subsection{Quantum Channel Discrimination}\label{sec:channel_discrimination}
	
	Let $\E$ and $\F$ be two quantum channels, each taking system $A$ to system $B$. In the discrimination setting we will generally use inputs to the channels that are entangled with a reference system. To simplify notation we will usually \emph{not} make the identity channel on the reference system explicit. Hence, if $\rho_{RA} \in \DM[R \otimes A]$ is a state, we will write
	\begin{equation}
		\E(\rho) \coloneqq \E_{A \to B}(\rho_{RA}) = (\id_{R} \otimes \E_{A \to B})(\rho_{RA}),
	\end{equation}
	and similarly also for $\F$. 
	
	\begin{figure}[htb]
		\centering
		\includegraphics[width=\linewidth]{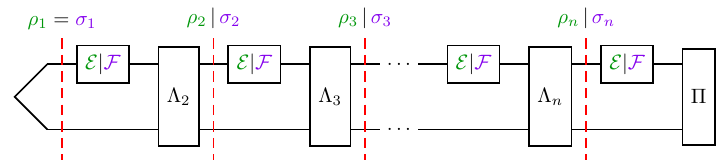}
		\caption{Illustration of a general adaptive protocol with $n$ uses of the black-box channel. The top row makes use of the given black-box $\E|\F$, which is either $\E$ or $\F$, while the bottom row depicts the memory system $R$. At various stages in the protocol, the green states $\rho$ occur if the channel is $\E$ and the purple states $\sigma$ occur if the channel is $\F$.   }
		\label{fig:adaptive_strategy}
	\end{figure}
	
	The most general channel discrimination protocol will choose input states based on the outputs of previous channel uses. This is called an adaptive protocol. 
	A general adaptive channel discrimination protocol with $n$ uses of the black-box channel ($\E$ or $\F$), can be fully specified by an initial state $\rho_1 = \sigma_1 \in \DM[R \otimes A]$, a set of $n - 1$ CPTP maps $\Lambda_i: R \otimes B \to R \otimes A$, that transform the state before it is fed into the next black-box channel, and a final binary POVM $\{\Pi, \IdentityMatrix -  \Pi\}$ on $R \otimes B$. We will assume the size of reference system $R$ to be fixed and identical throughout the protocol (this is without loss of generality). The protocol consists of alternating applications of the black-box channel and the preparation CPTP maps $\Lambda_i$ (see \autoref{fig:adaptive_strategy}). We define for $i \in \{2, \ldots, n\}$:
	\begin{equation}
		\rho_i \coloneqq \Lambda_{i}(\E(\rho_{i-1})), \quad \sigma_i  \coloneqq \Lambda_{i}(\F(\sigma_{i-1})),
	\end{equation}
	and so the final state before the action of the POVM will be $\E(\rho_n)$ if the channel is $\E$ and $\F(\sigma_n)$ if the channel is $\F$. 
	The distinguishability of the channels then comes down to the distinguishability of the two states $\E(\rho_n)$ and $\F(\sigma_n)$. We will be focussing again on the asymmetric setting, with the main object of study being the hypothesis testing relative entropy $D_H^\varepsilon(\E(\rho_n) \| \F(\sigma_n))$. We write the optimal adaptive type~II error rate for a given finite number of channel uses~$n$ as \cite{cooney_strong_2016}
	\begin{equation}
		\sup_{\rho_1, \{\Lambda_i\}_i} \frac{1}{n} D^\varepsilon_H(\E(\rho_n)\|\F(\sigma_n))
	\end{equation}
	where the supremum is over every initial input state $\rho_1 = \sigma_1$ and all subsequent input preparation CPTP maps $\{\Lambda_i\}_{i = 2}^n$.
	It was shown in \cite[Prop. 3]{katariya_evaluating_2021} that this is computable as a semi-definite program.
	
	\subsubsection{Asymptotic Equivalence of Adaptive and Parallel Strategies}
	
	It has been shown recently \cite{wilde_amortized_2020, wang_resource_2019} that asymptotically, when the number of channel uses goes to infinity, the best exponential decay rate of the type~II error (per channel use) such that the type~I error still goes to zero is given by the amortized Umegaki channel divergence, i.e.
	\begin{equation}
		\lim_{\varepsilon \to 0} \lim_{n \to \infty} \sup_{\rho_1, \{\Lambda_i\}_i} \frac{1}{n} D^\varepsilon_H(\E(\rho_n)\|\F(\sigma_n)) = D^A(\E\|\F)\,
		\label{eq:opt-adaptive-amortized}
	\end{equation}
	where
	\begin{equation}
		D^A(\E\|\F) \coloneqq \sup_{\rho, \sigma \in \DM[R \otimes A]}\big[D(\E(\rho)\|\F(\sigma)) - D(\rho\|\sigma)\big].
	\end{equation}
	Note that the dimension of the reference system $R$ in this last supremum can be arbitrarily large.
	
	The chain rule from \cite{fang_chain_2020} states that this amortized divergence is in fact equal to the regularized channel divergence, i.e.
	\begin{multline}
		D^A(\E\|\F) = D^{\text{reg}}(\E\|\F) \coloneqq \\\lim_{n \to \infty} \frac{1}{n} \sup_{\nu \in \DM[R^{\otimes n} \otimes A^{\otimes n}]} D(\E^{\otimes n}(\nu)\|\F^{\otimes n}(\nu)),
	\end{multline}
	the latter of which can be achieved by parallel protocols, as a consequence of \cite[Theorem 3]{wang_resource_2019}.
	Note that the reference system $R$ in the latter optimization can be chosen isomorphic to $A$.
	Hence, asymptotically (and in the regime in which the type~I error goes to zero) adaptive strategies offer no advantage over parallel ones. We provide a finite~$n$ version of this statement, by giving explicit bounds on how much the error probabilities of adaptive and parallel strategies can differ for a finite number of channel uses.
	
	\section{Parallelizing an \texorpdfstring{$n$}{n}-Shot Adaptive Protocol}
	\label{sec:main_result}
	
	We state our main result in two forms, first in a simple manner that illustrates the main idea and structure of the result, and secondly a more detailed theorem that gives a tighter bound, and which states in detail what one can choose as a parallel input state. 
	
	Remember that there is a tradeoff in minimizing the type~I and type~II errors. In the context of a strategy, for a given type~I error $\alpha$ we will write the best achievable type~II error as $\beta(\alpha)$. We are especially interested in the exponential decay rate of the type~II error with the number of channel uses; i.e., if some strategy involving $n$ uses of the channel has type~II error $\beta(\alpha)$, we are interested in the quantity $-\frac{1}{n} \log(\beta(\alpha))$. Our main result compares this error decay rate per channel use between adaptive and parallel strategies:
	
	\begin{corollary}[Main result, simple version]\label{cor:seq_par}
		Let $\E, \F: A \to B$ be two quantum channels such that $D_{\max}(\E\|\F) <  \infty$. Let there be an adaptive discrimination protocol with $n$ channel uses, that -- for an arbitrary type~I error $\alpha_a \in [0,1]$ -- achieves type~II error $\beta_a(\alpha_a)$. Then, for all $\alpha_p \in (0,1]$ there exists a parallel protocol with $m$ channel uses and type~II error $\beta_p(\alpha_p)$ such that for all $\alpha_a \in [0,1]$ the type~II error rates per channel use obey the following relation:
		\begin{multline}\label{eq:cor_seq_par}
			-\frac{1}{m} \log(\beta_p(\alpha_p)) \geq -\frac{1 - \alpha_a}{n} \log(\beta_a(\alpha_a)) \\- \frac{C n}{\sqrt{m}}\log(8\over \alpha_p) - \frac{1}{n}.
		\end{multline}
		That is, the type~II error rate of the parallel protocol is essentially at least as good as the adaptive one modulo an additional error term, which decays as $m \to \infty$. 
		The constant $C$ is given by
		\begin{equation}
			C \coloneqq 7 \log(2^{D_2(\E\|\F)} + 2) \leq 7 \log(2^{D_{\max}(\E\|\F)} + 2)\,.
		\end{equation}
	\end{corollary}
	
	\begin{remark}
		If we take the limit $m \to \infty$ in \eqref{eq:cor_seq_par}, then $n \to \infty$, and finally $\alpha_a \to 0$ and $\alpha_p \to 0$, we find that asymptotically there exist parallel strategies with better or at least equal type~II error decay than an adaptive one, and hence our result also implies the known result \cite{fang_chain_2020} that, asymptotically (and with $\alpha \to 0$), adaptive strategies offer no advantage over parallel ones. 
	\end{remark}
	
	\begin{remark}
		It is known that for small $n$ and $m$, and suitably chosen $\alpha_a$ and $\alpha_p$, the type~II error rates for adaptive and parallel strategies can be arbitrarily far apart (see \Cref{sec:example} below for an example). From the asymptotic equivalence, we know that this difference has to vanish as $n$ and $m$ go to infinity, but the purely asymptotic statement does not tell us how exactly this vanishes, and what the required relationship between $n$ and $m$ is (in principle the required $m$ to reach a similar rate as a given adaptive strategy with $n$ channel uses could grow arbitrarily fast with $n$).
		\Cref{cor:seq_par} now tells us that the difference of type~II error rates between an adaptive strategy and the corresponding parallel one will become arbitrarily small if $m = \omega(n^2)$ (i.e., $m$ has to grow faster than $n^2$). Hence, given a sequence of adaptive strategies with $n$ channel uses, we can convert these into parallel strategies using at most quadratically (or a little bit more than quadratically) as many channel uses each, and will achieve matching rates once $n$ gets large enough. This quadratic relationship is universal in the sense that it holds for all pairs of channels $\E$, $\F$ with $D_{\max}(\E\|\F) < \infty$ (where only the prefactor depends on the value of $D_{\max}(\E\|\F)$).
		
		Note again that we are not comparing type~II errors in \autoref{cor:seq_par}, but rather decay rates of the type~II error per channel use. When we say that the rates of a parallel strategy with $m = \Omega(n^2)$ channel uses and an adaptive strategy with $n$ channel uses are roughly equal, the parallel strategy will have much smaller type~II error because it has many more channel uses. 
	\end{remark}
	
	The following is the more refined version of our result, with a tighter bound and a description of the parallel input state:
	
	\begin{theorem}[Main result, technical version]\label{thm:seq_par}
		Let $\E, \F: A \to B$ be quantum channels such that $D_{\max}(\E\|\F) < \infty$. 
		Given an arbitrary adaptive protocol with $n$ channel uses, we write $\rho_i, \sigma_i \in \DM[R_a \otimes A]$, $i \in \{ 1, \ldots, n\}$ for the states that are input into the channel during the adaptive protocol ($\rho_i$ if the channel is $\E$, and $\sigma_i$ if the channel is $\F$; see \Cref{sec:channel_discrimination} and \autoref{fig:adaptive_strategy} above for a more detailed explanation of this notation). We define $\ell \in \{1, \ldots, n\}$ as the step in the protocol where the distinguishability increases the most, i.e.,
		\begin{equation}
			\ell \coloneqq  \argmax_{k\in \{ 1, \ldots, n\}} \Big[ D(\E(\rho_k)\|\F(\sigma_k)) - D(\rho_k\|\sigma_k) \Big]\,.
		\end{equation}
		
		Then, for all $\alpha_p \in (0,1]$, and $m \in \naturals$, there exists a state $\nu \in \DM[R^{\otimes m} \otimes A^{\otimes m}]$ such that for all $\alpha_a \in [0,1]$:
		\begin{multline}\label{eq:seq_par_simple}
			\frac{1}{m} D_H^{\alpha_p}\big(\E^{\otimes m}(\nu)\|\F^{\otimes m}(\nu)\big) \geq \frac{1 - \alpha_a}{n} D_H^{\alpha_a}\big(\E(\rho_n)\|\F(\sigma_n)\big) \\- \frac{c'_\ell}{\sqrt{m}} \qty[\log(4 \over \alpha_p) + K] -  \frac{1}{m}\qty[\log(1 \over \alpha_p) - \log(1 - \frac{\alpha_p}{4})] \\- \frac{h(\alpha_a)}{n}\, ,
		\end{multline}
		where 
		\begin{equation}
			K \coloneqq \frac{\ln(2) \log^2(3)}{8} \cosh(\log(3) \over 2) \leq 0.29
		\end{equation}
		and $c'_\ell$ depends on the pair of channels $\E$, $\F$ and can be bounded as follows:
		\begin{align}
			c'_\ell &\coloneqq {4 \over \log(3)} \inf_{{\gamma_1, \gamma_2 \in (0,1]}} \qty[c_{\gamma_1}(\E(\rho_\ell)\|\F(\sigma_\ell)) + c_{\gamma_2}(\rho_\ell\|\sigma_\ell)] \\
			&\leq {8 \ell \over \log(3)}  \inf_{\gamma \in (0,1]} \widehat{c}_\gamma(\E\|\F) \\
			&\leq {8 n \over \log(3)}  \inf_{\gamma \in (0,1]}  \widehat{c}_\gamma(\E\|\F)\,.
		\end{align}
		where 
		\begin{align}
			c_\gamma(\rho\|\sigma) &\coloneqq \frac{1}{\gamma}\log(2^{\gamma D_{1 + \gamma}(\rho\|\sigma)} + 2^{-\gamma D_{1 - \gamma}(\rho\|\sigma)} + 1) , \\
			\widehat{c}_\gamma(\E\|\F) &\coloneqq  \frac{1}{\gamma}\log(2^{\gamma \widehat{D}_{1 + \gamma}(\E\|\F)} + 2)\,.
		\end{align}
		\medskip
		
		\noindent Moreover, if 
		\begin{equation}\label{eq:thm_seq_par_bound_m}
			m \geq \log(4\over \alpha_p)\qty(\frac{4}{\log(3) \sqrt{2 \ln(2)}})^2 , 
		\end{equation}
		then we have the following tighter bound
		\begin{multline}\label{eq:thm_seq_par}
			\frac{1}{m} D_H^{\alpha_p}\big(\E^{\otimes m}(\nu)\|\F^{\otimes m}(\nu)\big) \geq \frac{1 - \alpha_a}{n} D_H^{\alpha_a}\big(\E(\rho_n)\|\F(\sigma_n)\big) \\- \frac{c_\ell}{\sqrt{m}} \sqrt{\log(4 \over \alpha_p)} -  \frac{1}{m}\qty[\log(1 \over \alpha_p) - \log(1 - \frac{\alpha_p}{4})] \\- \frac{h(\alpha_a)}{n}\, ,
		\end{multline}
		where $c_\ell$ is defined in a similar way as $c'_\ell$ but with different numerical constants $K_1$ and $K_2$:
		\begin{align}
			K_1 &\coloneqq  2 \sqrt{2 \ln(2) \cosh(\log(3) \over 2)} \leq 2.72, \\
			K_2 &\coloneqq 2 \sqrt{2 \ln(2)} \leq 2.36 ,  \\
			c_\ell & \coloneqq \inf_{{\gamma_1, \gamma_2 \in (0,1]}} \qty[ K_1 c_{\gamma_1}(\E(\rho_\ell)\|\F(\sigma_\ell)) + K_2 c_{\gamma_2}(\rho_\ell\|\sigma_\ell)]\label{eq:definition_cl}\\
			&\leq \ell (K_1 + K_2) \inf_{\gamma \in (0,1]} \widehat{c}_\gamma(\E\|\F) \\
			&\leq n (K_1 + K_2) \inf_{\gamma \in (0,1]}  \widehat{c}_\gamma(\E\|\F) \label{eq:upper_bound_cl} \,.
		\end{align}
		
		\noindent The parallel input state $\nu$ can be chosen as either:
		\begin{itemize}
			\item An optimizer in the smoothing of the max-divergence $D^{\varepsilon}_{\max}(\rho_l^{\otimes m} \|\sigma_l^{\otimes m})$, where $\varepsilon = \frac{1}{2}\left(1 - \sqrt{1 - \alpha_p}\right)$, i.e.,
			\begin{equation}
				\nu = \tilde{\nu}_{R_a^m A^m} \coloneqq \argmin_{\tilde{\rho} \in B^\circ_{\varepsilon}(\rho_\ell^{\otimes m})} D_{\max}(\tilde{\rho}\| \sigma_\ell^{\otimes m})\,.
			\end{equation}
			Note that the reference system $R_a$ depends on the adaptive protocol and can be arbitrarily large, and hence also $\tilde{\nu}$ might have an arbitrarily large reference system.
			
			\item The canonical (or any other) purification of the $A^m$-marginal $\tilde{\nu}_{A^m} = \Tr_{R_a^m}(\tilde{\nu}_{R_a^m A^m})$. That is, we can choose $\nu = \ketbra{\psi_{R^m A^m}}$, with
			\begin{equation}
				\ket{\psi_{R^m A^m}} = \IdentityMatrix_{R^m} \otimes \sqrt{\tilde{\nu}_{A^m}} \ket{\Phi_{R^m A^m}}
			\end{equation}
			where $R$ is isomorphic to $A$ and $\ket{\Phi_{R^m A^m}}$ is an unnormalized maximally entangled state. 
			
		\end{itemize}
	\end{theorem}
	
	\begin{remark}
		Even though the second choice for the parallel input state $\nu$ might seem preferable in most cases (due to the control over the size of the reference system $R$), it is not necessarily so. This is because even though the reference system of the first choice for $\nu$ could be very large, it is not always so. Since the overall state can be mixed, it might actually be smaller than the canonical purification of its marginal. Additionally, with the first choice for $\nu$, one gets a bound that this parallel input state is $\varepsilon$-close to the input state of the adaptive strategy, which might be useful in cases where one wants to show that some property of the adaptive strategy (e.g., a satisfied energy constraint for the input states) is (approximately) satisfied also for the parallel strategy. 
	\end{remark}
	
	\begin{remark}\label{remark:harrow_example}
		The constraint $D_{\max}(\E\|\F) < \infty$ is necessary in general, as the example of \cite{harrow_adaptive_2010} (see also \cite{salek_usefulness_2022}) serves as a counterexample to our statement without this constraint. Specifically, \cite{harrow_adaptive_2010} constructs two channels $\E$, $\F$ for which its authors then show that there exists an adaptive strategy with only two channel uses that achieves perfect discrimination, i.e., both $\alpha_a = 0$ and $\beta_a = 0$. In our terminology this implies that for all $\alpha_a \in [0,1]$
		\begin{equation}
			D_H^{\alpha_a}(\E(\rho_2)\|\F(\sigma_2)) = \infty\,.
		\end{equation}
		On the other hand, \cite{harrow_adaptive_2010} shows for these two channels that with a parallel strategy, even with arbitrarily many channel uses, perfect discrimination can never be achieved (i.e., $\alpha_p$ and $\beta_p$ cannot both be zero), which in our notation implies that for all $m$ and for all $\nu \in \DM[R \otimes A^{\otimes m}]$
		\begin{equation}
			D_H^0(\E^{\otimes m}(\nu)\|\F^{\otimes m}(\nu)) < \infty\,,
		\end{equation}
		Since it is well known (see, e.g., \cite[Prop.~4.66]{khatri_principles_2020}) that for all states $\rho, \sigma$
		\begin{equation}
			\lim_{\alpha_p \to 0} D_H^{\alpha_p}(\rho\|\sigma) = D_H^{0}(\rho\|\sigma)
		\end{equation}
		this means that for all $m$, one can find a sufficiently small $\alpha_p$ such that for all $\nu \in \DM[R \otimes A^{\otimes m}]$
		\begin{equation}
			D_H^{\alpha_p}(\E^{\otimes m}(\nu)\|\F^{\otimes m}(\nu)) < \infty.
		\end{equation}
		and thus a relation like \eqref{eq:thm_seq_par} cannot hold for these two channels. 
		Note that, even in this example, the Stein exponent, i.e., the optimal exponential decay rate of the type~II error such that the type~I error still goes to zero, is still identical for the parallel and adaptive strategies, as it is infinite also for the optimal parallel strategy.  
		This follows from
		\begin{equation}
			D_{\max}(\E\|\F) = \infty \quad \Rightarrow \quad D(\E\|\F) = D^{\mathrm{reg}}(\E \| \F) = \infty\,.
		\end{equation}
	\end{remark}
	
	\subsection{Computability}
	
	\begin{remark}\label{rem:computability}
		The usefulness of finite-length bounds often crucially depends on whether the quantities involved can actually be efficiently computed \cite{hayashi_finite-length_2020, watanabe_finite-length_2017, hayashi_uniform_2016}. To demonstrate this for our bound, we consider the following two applications:
		\begin{enumerate}
			\item We are given an adaptive strategy, where we know (potentially only a lower bound on) $D_H^{\alpha_a}(\E(\rho_n)\|\F(\sigma_n))$, and want to employ the theorem to get a lower bound on the performance of the best-possible parallel strategy (i.e., we want to know how much worse could a parallel strategy potentially be). 
			
			Such a lower bound can be computed using \autoref{thm:seq_par} in time $\mathcal{O}(1)$ (i.e., the computational complexity is independent of the number of channel uses $n$ and $m$), by using the bound on $c_{\ell}$ from \eqref{eq:upper_bound_cl}. A potentially much tighter bound, obtained by finding $\ell$ explicitly and then calculating $c_\ell$, can be computed in time~$\mathcal{O}(n)$. 
			
			\item We want to upper bound the performance of the best possible adaptive strategy with $n$ channel uses, by finding the best possible parallel strategy and then using \autoref{thm:seq_par}. We show in \autoref{lem:parallel_polytime} below that the best possible parallel strategy with $m$ channel uses can actually be calculated in $\mathcal{O}(\poly(m))$ time, and hence by choosing $m = n^{2 + \xi}$ for $\xi > 0$ we also obtain asymptotically tight upper bounds on any adaptive strategy with $n$ channel uses in $\mathcal{O}(\poly(n))$ time.
			
			To our knowledge this is the first such poly-time bound obtained in the literature.
			Specifically, computing the best parallel strategy and then using our bound is exponentially faster than any currently known way of optimizing over all adaptive strategies. While the authors of \cite{katariya_evaluating_2021} showed that optimizing over all adaptive strategies can be phrased as an SDP, the size of this SDP grows exponentially in $n$, and there is no obvious symmetry (like the permutation invariance in the parallel case) that allows one to reduce the number of variables.
		\end{enumerate}
	\end{remark}
	\renewcommand{\O}{\mathcal{O}}
	\begin{lemma}\label{lem:parallel_polytime}
		Let $\E$, $\F: A \to B$ be two quantum channels such that $D_{\max}(\E\|\F) < \infty$. 
		Then, for a fixed $\varepsilon \in [0,1)$, the quantity 
		\begin{equation}\label{eq:computability_result}
			\frac{1}{n} D_H^{\varepsilon}(\E^{\otimes n}\|\F^{\otimes n}) = \frac{1}{n} \sup_{\nu \in \DM[R^{\otimes n} \otimes A^{\otimes n}]} D_H^{\varepsilon}(\E^{\otimes n}(\nu)\|\F^{\otimes n}(\nu))
		\end{equation}
		can be computed up to an additive error $\delta>0$ in time $\mathcal{O}(\poly(n)\log(\frac{1}{\delta}))$ as $n \to \infty$ and $ \delta \to 0$.
	\end{lemma}
	\begin{proof}
		
		As shown in \cite[Prop. 2]{wang_resource_2019}, the quantity $2^{-D_H^{\varepsilon}(\E^{\otimes n}\|\F^{\otimes n})}$ can be expressed as the following semidefinite program (SDP) over operators $\Omega_{R^nB^n} \in \BO[\HS_R\n \otimes \HS_B\n]$ and  $    \rho_{R^n} \in \BO[\HS_R\n]$
		\begin{mini}|l|
			{\Omega_{R^n B^n}, \rho_{R^n}}{\Tr(\Omega_{R^n B^n} \Gamma^{\F^{\otimes n}}_{R^n B^n})}
			{}{}
			\addConstraint{\Tr(\Omega_{R^n B^n} \Gamma^{\E\n}_{R^n B^n})}{ \geq 1 - \varepsilon}
			\addConstraint{0 \leq \Omega_{R^n B^n}}{ \leq \rho_{R^n} \otimes \IdentityMatrix_{B^n}}
			\addConstraint{\Tr(\rho_{R^n})}{=1},
			\label{eq:sdp_original}
		\end{mini}
		where $\Gamma^{\E}_{R B}$ is the Choi matrix of $\E$ and it is easy to see that $\Gamma^{\E^{\otimes n}}_{R^n B^n} = (\Gamma^{\E}_{R B})^{\otimes n}$. The operators $\Omega_{R^n B^n}$ and $\rho_{R^n}$ in this SDP have a number of parameters exponential in $n$, but we will show that we can use the permutation symmetry of this problem to rephrase it as an SDP polynomial in~$n$. Throughout this proof we use $\HS$ to denote a general Hilbert space,  which we will then choose to be either $\HS_R$, $\HS_B$, or $\HS_R \otimes \HS_B$ in different situations. We will denote any objects that depend on the chosen Hilbert space with a subscript or superscript $\HS$ (such as $P_{\HS}$ in the following paragraph), where we then replace the subscript or superscript with $R$, $B$, or $RB$ whenever we choose a specific Hilbert space; e.g., we write $P_{R}$, $P_B$, or $P_{RB}$.
		
		For any permutation $\pi \in \Sn$ (where $\Sn$ is the symmetric group) we write $P_{\HS}(\pi)$ for the permutation matrix corresponding to the action of $\pi$ on $\HS^{\otimes n}$ by permuting the $n$ copies of $\HS$. It is then easy to see that these permutation matrices are unitary and $P_{\HS}(\pi)^\dagger = P_{\HS}(\pi^{-1})$. For any operator $X \in \BO[\HS^{\otimes n}]$ we also write the group average as
		\begin{equation}
			\overline{X} \coloneqq {1 \over \abs{\Sn}} \sum_{\pi \in \Sn} P_\HS(\pi) X P_\HS(\pi)^{\dagger}\, ,
		\end{equation}
		and the set of all permutation invariant operators as
		\begin{multline}
			\End[\HS\n] \coloneqq \\ \{A \in \BO[\HS\n] \,|\, P_{\HS}(\pi) A P_{\HS}(\pi)^{\dagger} = A, \, \forall \pi \in \Sn\} \,.
		\end{multline}

		We start by showing that the minimum in our SDP \eqref{eq:sdp_original} is always achieved by permutation invariant operators. Let $(\Omega_{R^nB^n},\rho_{R^n})$ be feasible for the optimization problem; i.e., they satisfy the constraints of \eqref{eq:sdp_original}. Let $\pi \in \Sn$ be any permutation; then
		\begin{equation}
			\Tr(P_{R}(\pi)\, \rho_{R^n} \, P_{R}(\pi)^\dagger) = \Tr(\rho_{R^n})
		\end{equation}
		since the $P_{R}(\pi)$ are unitary, and thus also $\Tr(\rho_{R^n}) = \Tr(\overline{\rho_{R^n}})$. Similarly, we get
		\begin{equation}
			0 \leq P_{R B}(\pi)\,\Omega_{R^n B^n}\,P_{RB}(\pi)^{\dagger} \leq P_{R}(\pi)\rho_{R^n}P_{R}(\pi)^\dagger \otimes \IdentityMatrix_{B^n}
		\end{equation}
		since positivity is preserved under unitary conjugation, and we also used 
		\begin{equation}
			P_{R B}(\pi) = (\IdentityMatrix_{R^n} \otimes P_{B}(\pi))(P_{R}(\pi) \otimes \IdentityMatrix_{B^n})
		\end{equation}
		which is immediate from the definition. This again implies that
		\begin{equation}
			0 \leq \overline{\Omega_{R^n B^n}} \leq \overline{\rho_{R^n}} \otimes \IdentityMatrix_{B^n}\,.
		\end{equation}
		By the cyclicity of the trace, and the fact that         
		$\Gamma^{\E^{\otimes n}}_{R^n B^n} = (\Gamma^{\E}_{R B})^{\otimes n}$, we also see that 
		\begin{align}
			\Tr(\overline{\Omega_{R^nB^n}}\Gamma^{\E^{\otimes n}}_{R^n B^n}) &= \Tr(\Omega_{R^nB^n}\Gamma^{\E^{\otimes n}}_{R^n B^n}) \, , \\
			\Tr(\overline{\Omega_{R^nB^n}}\Gamma^{\F^{\otimes n}}_{R^n B^n}) &= \Tr(\Omega_{R^nB^n}\Gamma^{\F^{\otimes n}}_{R^n B^n})\,.
		\end{align}        
		Hence, also $(\overline{\Omega_{R^nB^n}}, \overline{\rho_{R^n}})$ are feasible for the optimization problem (i.e., they satisfy the constraints) and also achieve the exact same value as $(\Omega_{R^nB^n},\rho_{R^n})$ does.
		The group averages are elements of $\End[R^n B^n]$ and $\End[R^n]$ respectively, and hence we can restrict the optimization in~\eqref{eq:sdp_original} to such permutation invariant operators.
		
		While the dimension of the subspace of these permutation invariant operators is polynomial in $n$, this does not yet show computability in $\poly(n)$ time using standard SDP solvers, as the problem is not phrased using matrices of size $\poly(n)$. In order to rephrase our problem in this way we follow the approach of \cite{fawzi_hierarchy_2022}, where a similar statement has been shown for the sharp \renyi divergence $D_{\alpha}^{\#}(\E\n\|\F\n)$. The key idea is to construct a suitable basis of the permutation invariant subspaces and then rephrase the SDP as one in which we minimize over (now only $O(\poly(n))$ many) basis coefficients. 
		As explained in \cite[page~7352]{fawzi_hierarchy_2022} (see also \cite{de_klerk_reduction_2007, anjos_handbook_2012}), one can construct such an orthogonal basis $C_r^{\HS}$, $r \in \{ 1, \ldots, m_{\HS}\}$ of $\End[\HS\n]$ (where orthogonality is with respect to the Hilbert-Schmidt inner product, and $m_{\HS} = \dim \End[\HS\n]$) as follows:
		Let $\{\ket{i}\}_{i = 1}^{d_{\HS}}$ be an orthonormal basis of $\HS$. Then, for any multi-index $\mathbf{i} \in \{1, \ldots, d_{\HS}\}^n$ define the associated vector
		$\ket{\mathbf{i}} = \bigotimes_{k = 1}^n \ket{i_k}$ on the tensor-product system $\HS\n$, and it is immediate that the set of all these vectors forms an orthonormal basis of $\HS^{\otimes n}$. 
		For any pair of such multi-indices $(\mathbf{i},\mathbf{j})$ -- this pair should be thought of as indexing a matrix element of an operator on $\HS\n$ -- we write the group orbit of these indices under the action of $\Sn$ as 
		\begin{equation}
			O(\mathbf{i},\mathbf{j}) \coloneqq \Set{(\pi(\mathbf{i}), \pi(\mathbf{j})) | \pi \in \Sn }\,,
		\end{equation}
		where $\pi(\mathbf{i})$ permutes the components of $\mathbf{i}$, i.e., $\pi(\mathbf{i})_k = \mathbf{i}_{\pi^{-1}(k)}$ for $k \in \{ 1, \ldots, n\}$.
		There are exactly $m_{\HS} = \dim \End[\HS\n]$ such group orbits, which we will label as $O_r^{\HS}$, $r \in \{ 1, \ldots, m_{\HS}\}$, and a representative element of each such orbit (i.e., a pair $(\mathbf{i}, \mathbf{j}) \in O_r^{\HS}$) can be efficiently computed given $r$, as shown in \cite[page~7352]{fawzi_hierarchy_2022}. This corresponds to the intuition that for $A \in \End[\HS\n]$ we require $A_{\mathbf{i}\mathbf{j}} = A_{\pi(\mathbf{i})\pi(\mathbf{j})}$ for any $\pi \in \Sn$; hence the matrix elements of $A$ have to be constant on each group orbit, and so the number of group orbits is equal to the dimension of $\End[\HS\n]$. 
		The basis elements $C_r^{\HS} \in \End[\HS\n]$ are now defined by specifying their matrix elements as
		\begin{equation}
			(C_r^{\HS})_{\mathbf{i}\mathbf{j}} \coloneqq \begin{cases} 1 & \text{if } (\mathbf{i},\mathbf{j}) \in O_r^{\HS} \\ 0 & \text{otherwise.} \end{cases} \qquad r = 1, \ldots, m_{\HS}. 
		\end{equation}
		It follows immediately from the definition that the $C_r^{\HS}$ are orthogonal, and since we have $m_{\HS} = \dim \End[\HS\n]$ of them, they form a basis.
		
		As explained in \cite[page~7353]{fawzi_hierarchy_2022}, for any matrix $A^{\otimes n} \in \End[\HS\n]$, its coefficients with respect to the basis $C_r^{\HS}$ can be computed straightforwardly from a description of $A$ by picking a representative of each orbit, and these coefficients are hence computable in $\O(\poly(n))$ time. Specifically, one can show that for any $r$ and any pair of indices in the group orbit $(\mathbf{i},\mathbf{j}) \in O^{\HS}_r$, the corresponding basis coefficient is given by
		\begin{equation}
			\gamma_r \coloneqq \prod_{k = 1}^n A_{\mathbf{i}_k \mathbf{j}_k}\, ;
		\end{equation}
		i.e., we get
		\begin{equation}
			A\n = \sum_{r=1}^{m_{\HS}} \gamma_r C^{\HS}_r\,.
		\end{equation}
		This can be applied to the Choi matrix $\Gamma^{\E^{\otimes n}}_{R^n B^n} = (\Gamma^{\E}_{R B})^{\otimes n}$, and hence we write
		\begin{equation}\label{eq:sdp_choi_coefficient}
			(\Gamma^{\E}_{RB})\n = \sum_{r = 1}^{m_{RB}} \gamma^{\E}_r C_r^{RB}
		\end{equation}
		and similarly when $\E$ is replaced by $\F$. 
		
		Additionally, for any finite-dimensional Hilbert space $\HS$, there exists a $*$-algebra isomorphism from $\End[\HS\n]$ to block-diagonal matrices \cite[Theorem 1]{gijswijt_matrix_2005}
		\begin{equation}\label{eq:sdp_star_isomorphism}
			\phi_{\HS}: \End[\HS\n] \to \bigoplus_{i = 1}^{t_{\HS}} \complex^{m_i \cross m_i}\,,
		\end{equation}
		where
		\begin{align}
			t_{\HS} &\leq (n + 1)^{d_{\HS}}, \\
			d_{\HS} &= \dim (\HS), \\
			\sum_{i = 1}^{t_{\HS}} m_i^2 &= \dim(\End[\HS\n]) \leq (n + 1)^{d_{\HS}^2}\,. \label{eq:sdp_star_isomorphism_coefficients}
		\end{align}
		Introducing the notation $M_{\HS} \coloneqq \sum_{i = 1}^{t_{\HS}} m_i$, for any $A \in \End[\HS\n]$ we have that $\phi_{\HS}(A) \in \complex^{M_{\HS} \cross M_{\HS}}$, and we write
		$\llbracket \phi_{\HS}(A) \rrbracket_i$ for the $i$-th block of $\phi_{\HS}(A)$. Crucially, by \cite[Lemma 3.3]{fawzi_hierarchy_2022}, \cite[Prop. 2.4.4]{polak_new_2019}, 
		\begin{equation}\label{eq:end_positivity}
			A \geq 0 \Leftrightarrow \phi_{\HS}(A) \geq 0 \Leftrightarrow \llbracket \phi_{\HS}(A) \rrbracket_i \geq 0 \; \forall i \in \{ 1, \ldots, t_{\HS} \} \,.
		\end{equation}
		From \cite[Theorem 1]{gijswijt_matrix_2005} it also follows that $\phi_{\HS}$ preserves orthogonality; i.e., if $\Tr(A^{\dagger} B) = 0$ then also $\Tr(\phi_{\HS}(A)^{\dagger} \phi_{\HS}(B)) = 0$.
		Remember that $\{C_r^{\HS}\}_{r \in \{1, \ldots, m_{\HS}\}}$ is a basis for $\End[\HS\n]$, and hence we can expand $\Omega_{R^n B^n} \in \End$ and $\rho_{R^n} \in \End[R^n]$ as follows:
		\begin{align}
			\Omega_{R^n B^n} &= \sum_{r=1}^{m_{RB}} y_r C_r^{RB}, \\
			\rho_{R^n} &= \sum_{r = 1}^{m_R} z_r C_r^{R}\,,
		\end{align}
		where $\{y_r \in \complex\}_{r=1}^{m_{RB}}$ and $ \{z_r \in \complex\}_{r=1}^{m_{R}}$ are the respective basis coefficients. Note that since the $C_r$ are not necessarily Hermitian, the coefficients are not necessarily real.
		Since optimizing over elements is obviously equivalent to optimizing over their basis coefficients, we can rephrase our SDP as follows:
		\begin{mini}|l|[1]
			{\{y_r\}_{r=1}^{m_{RB}}, \{z_r\}_{r=1}^{m_{R}}}{\sum_{r=1}^{m_{RB}} y_r (\gamma_r^{\F})^* \Tr((C_r^{RB})^\dagger C_r^{RB})}
			{}{}
			\label{eq:sdp_reduced}
			\addConstraint{\sum_{r=1}^{m_{RB}} y_r (\gamma_r^{\E})^{*} \Tr((C_r^{RB})^\dagger C_r^{RB}) \geq 1 - \varepsilon}
			\addConstraint{\sum_{r=1}^{t_R}z_r \Tr(C_r^R)=1\,}{}
			\addConstraint{0 \leq \sum_{r = 1}^{m_{RB}} y_r\llbracket\phi_{RB}(C_r^{RB})\rrbracket_i }
			\\&\quad &&\phantom{0} \leq \sum_{r=1}^{m_R} z_r \llbracket \phi_{RB}(C_r^{R} \otimes I_{B^n}) \rrbracket_i\quad \forall i \in \{1, \ldots, t_{RB}\}
		\end{mini}
		where we also used \eqref{eq:end_positivity} and \eqref{eq:sdp_choi_coefficient}.
		
		We will show that this is an SDP in $\poly(n)$ variables with $\poly(n)$ constraints by casting it into standard form.  
		\DeclarePairedDelimiter\hs{\langle}{\rangle}
		To simplify notation we will write the Hilbert-Schmidt inner product using $\hs{\cdot, \cdot}$, i.e., $\hs{A, B} \coloneqq \Tr(A^\dagger B)$ and also write $M = M_{RB}$ (i.e., $M = M_{\HS}$ as below \eqref{eq:sdp_star_isomorphism_coefficients} with $\HS = \HS_R \otimes \HS_B$). For $s \in \complex$, consider the following $(2 M + 1) \cross (2M + 1)$ block-diagonal matrix
		\newcommand{\diag}{\mathrm{diag}}
		\begin{multline}
			X \coloneqq \qty[\sum_{r = 1}^{m_{RB}} y_r\phi_{RB}(C_r^{RB})]  \\ \oplus \qty[\sum_{r'=1}^{m_R} z_{r'} \phi_{RB}(C_{r'}^{R} \otimes I_{B^n}) - \sum_{r = 1}^{m_{RB}} y_r\phi_{RB}(C_r^{RB})] \oplus s \label{eq:sdp_def_standard_form}
		\end{multline}
		where $s$ is a slack variable that turns the one inequality constraint into an equality constraint (see further below). 
		From here on, we write $(\,\cdot\,) \oplus (\,\cdot\,) \oplus (\,\cdot\,)$ to specify a block-diagonal matrix with block sizes equal to the ones in \eqref{eq:sdp_def_standard_form}.
		The constraint $X \geq 0$ now implies 
		\begin{equation}
			0 \leq \sum_{r = 1}^{m_{RB}} y_r\llbracket\phi_{RB}(C_r^{RB})\rrbracket_i \leq \sum_{r=1}^{m_R} z_r \llbracket \phi_{RB}(C_r^{R} \otimes I_{B^n}) \rrbracket_i
		\end{equation}
		for $i = 1, \ldots, t_{RB}$.
		Additionally, since $\phi_{RB}$ preserves orthogonality, we can recover the coefficients $y_r$ and $z_r$ from $X$ by taking inner products with suitable operators. Specifically, with
		\begin{align}
			\tilde{Y}_r &\coloneqq \phi_{RB}(C_r^{RB}) \oplus 0 \oplus 0\, , \\
			Y_r &\coloneqq \frac{\tilde{Y}_r}{\hs{\tilde{Y}_r, \tilde{Y}_r}} \, ,
		\end{align}
		for $r \in \{1, \ldots, m_{RB}\}$, and with
		\begin{align}
			\tilde{Z}_r &\coloneqq \phi_{RB}(C_r^{R} \otimes \IdentityMatrix_{B^n}) \oplus \phi_{RB}(C_r^R \otimes \IdentityMatrix_{B^n}) \oplus 0 \, ,\\
			Z_r &\coloneqq \frac{\tilde{Z}_r}{\hs{\tilde{Z}_r, \tilde{Z}_r}}\, ,
		\end{align}
		for $r \in \{ 1, \ldots, m_{R}\}$, we have 
		\begin{align}
			y_r & = \hs{Y_r, X}, \\
			z_r & = \hs{Z_r,  X}.
		\end{align}
		Hence we can rephrase the expressions in \eqref{eq:sdp_reduced} as inner products with $X$. Specifically, 
		\begin{multline}
			\sum_{r=1}^{m_{RB}} y_r (\gamma_r^{\E})^* \Tr((C_r^{RB})^\dagger C_r^{RB}) = \\ \hs*{\sum_{r=1}^{m_{RB}} Y_r (\gamma_r^{\E})^* \Tr((C_r^{RB})^\dagger C_r^{RB}), X}
		\end{multline}
		and one can do the same with $\E$ replaced by $\F$.
		
		We want to transform our SDP into one where we optimize over $X$, and for that we need to impose the necessary block-diagonal structure of $X$ (including the block-diagonal substructure that comes from its parts being images of $\phi_{RB}$). For this, consider the following linear space 
		\begin{multline}
			\Big\{\,\qty[\phi_{RB}(\Omega)] \oplus \qty[\phi_{RB}(\rho \otimes \IdentityMatrix_{B^n}) - \phi_{RB}(\Omega)] \oplus s \;\Big|\; \\ \Omega \in \End,\, \rho \in \End[R^n],\, s \in \complex\,\Big\} \label{eq:sdp_matrix_subspace}
		\end{multline}
		and define $\mathcal{A}$ to be a set of matrices that form a basis of the orthogonal complement of this linear space, where we take the orthogonal complement within $\complex^{(2M + 1) \cross (2M + 1)}$. It is then easy to see that $|\mathcal{A}| \leq \dim(\complex^{(2M + 1) \cross (2M + 1)}) = (2M + 1)^2 = \O(\poly(n))$.
		
		With $S \coloneqq (0 \,\oplus\, 0 \,\oplus\, 1)$, we can then introduce the following SDP
		\begin{mini}|l|[1]
			{X}{\hs*{\sum_{r=1}^{m_{RB}} Y_r (\gamma_r^{\F})^* \Tr((C_r^{RB})^\dagger C_r^{RB}), X}}{}{}
			\label{eq:sdp_standard_form}
			\addConstraint{\hs*{\qty(\sum_{r=1}^{m_{RB}} Y_r (\gamma_r^{\E})^{*} \Tr((C_r^{RB})^\dagger C_r^{RB})) - S, X} = 1 - \varepsilon}{}
			\addConstraint{X \in \complex^{(2M + 1) \cross (2M + 1)}}
			\addConstraint{X \geq 0}{}
			\addConstraint{\hs*{\sum_{r=1}^{t_R}Z_r \Tr(C_r^R), X} = 1}{}
			\addConstraint{\hs{A, X} = 0 \qquad \forall A \in \mathcal{A}}
		\end{mini}
		which is in standard form with $\O(\poly(n))$ many constraints and matrices of size $\O(\poly(n))$.
		
		This new SDP is equivalent to \eqref{eq:sdp_reduced}, which can be seen as follows: From \eqref{eq:sdp_def_standard_form} and the calculations thereafter it follows that for every feasible $(\{y_r\}_{r = 1}^{m_{RB}}, \{z_r\}_{r=1}^{m_R})$ in \eqref{eq:sdp_reduced} there is a corresponding feasible $X$ in \eqref{eq:sdp_standard_form} which achieves the same value. 
		Conversely, every $X$ that satisfies the constraints of \eqref{eq:sdp_standard_form} has to lie in the subspace \eqref{eq:sdp_matrix_subspace}, and it is then immediate that it can be represented as in \eqref{eq:sdp_def_standard_form} for suitable $y_r$ and $z_r$. These $y_r$ and $z_r$ then also achieve the same value in \eqref{eq:sdp_reduced}  as $X$ did in \eqref{eq:sdp_standard_form}.  
		
		It now only remains to show that we can also efficiently calculate all that is required to parameterize this SDP. It was shown in \cite[pages 7352-7353]{fawzi_hierarchy_2022} (see also \cite{litjens_semidefinite_2017}) that the $\llbracket \phi_{\HS}(C_r^{\HS}) \rrbracket_i$ (and hence also the $\phi_{\HS}(C_r^{\HS})$) can be computed in $\poly(n)$ time. As $I_{B^n} = I_B^{\otimes n}$, the coefficients of $I_{B^n}$ with respect to the $\{C_r^{B}\}$ can be computed efficiently just as for the Choi matrices above, and additionally, from the construction of the $C_r^{\HS}$, there is an obvious one-to-one mapping $C_r^R \otimes C_{r'}^B = C_{r''}^{RB}$; hence also $\phi_{RB}(C_r^{R} \otimes I_{B^n})$ is efficiently computable. This gives a generating set of the subspace \eqref{eq:sdp_matrix_subspace} and a basis for its orthogonal complement is then also computable in $\O(\poly(n))$ time. 
		Finally, from the construction of the $C_r^{\HS}$, it is immediate that the norm $\Tr((C_r^{\HS})^\dagger C_r^{\HS})$ is given by the size of the group orbit, which is equal to the number of distinct permutations of an element $(\mathbf{i}, \mathbf{j})$ in the orbit. Concretely, if $\ell \in \{1, \ldots, d_{\HS}^2\}$ indexes all pairs of single-system indices $(i,j)$, $i,j \in \{1, \ldots, d_{\HS}\}$, then given $(\mathbf{i}, \mathbf{j}) \in O_r^{\HS}$, let $K_{\ell}$ denote the number of occurrences of $\ell = (i,j)$ in the multi-index $(\mathbf{i}, \mathbf{j})$, i.e., the number of different $k$ values ($k \in \{1, \ldots, n\}$) for which $\mathbf{i}_k = i$, $\mathbf{j}_k = j$. The size of the group orbit is then given by
		\begin{equation}
			\Tr((C^{\HS}_r)^\dagger C^{\HS}_r) = |O_r^{\HS}| = \binom{n}{K_1, \ldots, K_{d_{\HS}^2}}
		\end{equation}
		and this multinomial coefficient can be calculated in time $\O(d_{\HS}^2)$; see, e.g., \cite{araujo_fast_2021}. 
		
		Also, $\Tr(C_r^R)$ can be efficiently computed by picking a representative $(\mathbf{i}, \mathbf{j})$ of the orbit $O_r^{R}$ and counting the number of different $k$ where $\mathbf{i}_k = \mathbf{j}_k$.
		Thus, we can compute all the coefficients in our SDP in $\poly(n)$ time. 
		An SDP with $\poly(n)$ variables and $\poly(n)$ constraints can be solved up to an additive error $\delta'$ in time $\mathcal{O}(\poly(n) \log(1/\delta'))$; see, e.g., \cite{jiang_faster_2020}.
		The final step in our proof is to pick a suitable $\delta'$ so that we can solve our original problem \eqref{eq:computability_result} up to an additive error~$\delta$.
		For any $n$ we write the solution for our SDP as
		\begin{equation}
			R_n \coloneqq 2^{-D_H^{\varepsilon}(\E^{\otimes n}\|\F^{\otimes n})}\,.
		\end{equation}
		Let us start by showing that $- \log(R_n) = \O(n)$ as $n \to \infty$.
		By \autoref{lem:anshu_min_max}, we have for all $\gamma \in (0, 1 - \varepsilon)$ that
		\begin{align}
			&D_H^{\varepsilon}(\E\n\|\F\n) \notag \\
			&\leq D_{\max}^{\sqrt{1 - \varepsilon - \gamma}}(\E\n\|\F\n) + \log(4 (\varepsilon + \gamma) \over \gamma^2) \\ 
			& \leq D_{\max}(\E\n\|\F\n) + \log(4 (\varepsilon + \gamma) \over \gamma^2) \\ &= n D_{\max}(\E\|\F) + \log(4 (\varepsilon + \gamma) \over \gamma^2) = \O(n)\,.
		\end{align}
		Note that the case $\varepsilon = 0$ (while not directly covered by \autoref{lem:anshu_min_max}) also follows immediately by first using that $D_H^0(\E\|\F) \leq D_H^{\varepsilon'}(\E\|\F)$ for all $\varepsilon' \in (0,1)$.
		
		Now, for any given $\delta > 0$, let us pick $\delta' \coloneqq (2^{\delta} - 1) R_n$. Remember that $\delta'$ was the additive error we make when solving the SDP \eqref{eq:sdp_reduced}, and this error then propagates to our original problem via
		\begin{align}
			-\frac{1}{n}\log(R_n + \delta') &= -\frac{1}{n} \log(R_n 2^{\delta}) \\
			&= -\frac{1}{n} \log(R_n) - \frac{\delta}{n} \\
			&= \frac{1}{n} D_H^{\varepsilon}(\E\n\|\F\n) - \frac{\delta}{n}
		\end{align}
		and hence this leads to an additive error of $\delta/n \leq \delta$ for the original problem.\footnote{This last inequality might seem very far from optimal, and if one wants to make further assumptions on how $n \delta$ behaves as $n \to \infty$, $\delta \to 0$ our runtime bounds can indeed be improved; however, this is not required for what we desire to show. Also, the signs here are chosen as the errors appear in practice: the SDP \eqref{eq:sdp_reduced} is a minimization, so any numerical approximation will be larger than the true value and the additive error hence positive, which then translates to a value smaller than the true value for our original problem \eqref{eq:computability_result}.}
		It is easy to see that $\log(2^\delta - 1) = \log(\delta) + \O(1)$ as $\delta \to 0$, which implies $\log(1/\delta') = -\log(R_n) - \log(2^\delta - 1) = \O(n) + O(\log(1/\delta)) = O(n \log(1/\delta))$ and hence our original problem can be calculated up to an error $\delta$ in time $\O(\poly(n) \log(1/\delta')) = \O(\poly(n) \log(1/\delta))$. 
	\end{proof}
	

	\subsection{A Simple One-Shot Version of the Chain Rule}
	
	The remainder of this section proves \autoref{thm:seq_par}. The general idea of the proof is the following: We will start by moving from the hypothesis testing relative entropy of the adaptive strategy to the Umegaki relative entropy, using \autoref{lem:wang_renner_upper_bound}. Then, we will see that by using an amortization argument (equations \eqref{eq:amort-1}--\eqref{eq:amort-last} below), we can bound the performance of the adaptive strategy by the performance of just one of its steps, the step where the distinguishability of the two states (that occur if the channel is either $\E$ or $\F$) increases the most. We then consider $m$ parallel copies of this step and construct a parallel input state using a chain rule for the smoothed max-relative entropy (\autoref{lem:chain_rule_simple}); see \autoref{fig:construction_parallel_input_state} below for an illustration of this step. The smoothed max-relative entropies can be related to the Umegaki relative entropies we used in the amortization argument by the non-asymptotic bounds presented in \autoref{lem:aep_finite_n}. Finally we connect to the hypothesis testing relative entropy of a parallel protocol using \autoref{lem:anshu_min_max}.  
	
	\begin{figure}[ht]
		\centering
		\includegraphics[width=\linewidth]{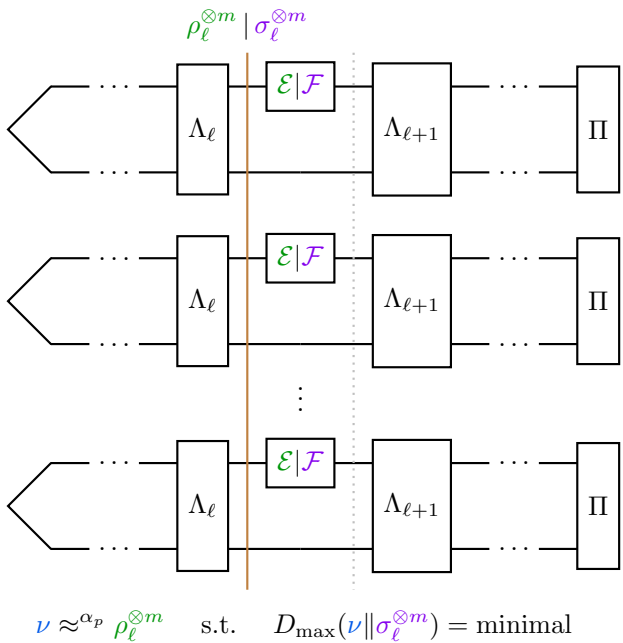}
		\caption{Illustration of a key step in our proof, the construction of the parallel input state. We start by picking a single step $\ell \in \{1, \ldots, n\}$ out of the adaptive protocol, where the distinguishability increase $D(\E(\rho_\ell)\|\F(\sigma_\ell)) - D(\rho_\ell\|\sigma_\ell)$ is maximal (this corresponds to the step from the orange to the dotted grey line in the diagram). Now consider $m$ copies of the adaptive strategy in parallel. We construct our parallel input state $\nu$ starting from $m$ copies of the input state of the adaptive strategy at this step $\ell$ if the channel was $\E$ (this is $\rho_\ell^{\otimes m})$. The state $\nu$ is then smoothed a bit to reduce its distance to $\sigma_\ell^{\otimes m}$ (which is the input state that we would have if the channel was $\F$). The degree to which we smooth depends on the type~I error $\alpha_p$ we want to achieve with the parallel strategy. Having a small type~I error means that the state~$\nu$ is very close to $\rho_\ell^{\otimes m}$, whereas allowing for a larger type~I error will move the state closer to $\sigma_\ell^{\otimes m}$.}
		\label{fig:construction_parallel_input_state}
	\end{figure}

	This subsection and the following one present lemmas we will use as part of our proof of \autoref{thm:seq_par}.
	
	The following is a variant of \cite[Prop.~3.2]{fang_chain_2020}, where our different convention of smoothing the max-divergence (smoothing only over normalized states) leads to a tighter and simpler bound, and restricting to the case of a single input system also makes things a bit simpler.
	\begin{lemma}\label{lem:chain_rule_simple}
		Let $\E$ and $\F$ be arbitrary quantum channels from system $A$ to system $B$, and let $\rho \in \DM[A]$, $\sigma \in \pos[A]$. Then for all $\varepsilon, \varepsilon' \in [0,1]$ there exists a state $\nu \in B_\varepsilon^\circ(\rho)$ such that
		\begin{equation}\label{eq:chain_rule_simple}
			D^{\varepsilon + \varepsilon'}_{\max}(\E(\rho)\|\F(\sigma)) \leq  D^{\varepsilon}_{\max}(\rho\|\sigma) + D^{\varepsilon'}_{\max}(\E(\nu)\|\F(\nu))\,.
		\end{equation}
		Moreover, $\nu$ can be chosen as
		\begin{equation}
			\nu = \argmin_{\tilde{\rho} \in B^\circ_{\varepsilon}(\rho)} D_{\max}(\tilde{\rho}\|\sigma) \,.
		\end{equation}
	\end{lemma}
	
	\begin{proof}
		Let $\nu \in B^\circ_\varepsilon(\rho)$ be an optimal choice for $D_{\max}^\varepsilon(\rho\|\sigma)$; i.e.,
		\begin{equation}
			\nu \leq 2^{D_{\max}^\varepsilon(\rho\|\sigma)} \sigma .
		\end{equation}
		Since $\F$ is a positive map, this implies that
		\begin{equation}
			\F(\nu)  \leq 2^{D_{\max}^\varepsilon(\rho\|\sigma)} \F(\sigma)\,.
		\end{equation}
		Furthermore, let $\tau \in B^\circ_{\varepsilon'}(\E(\nu))$ be an optimal choice for $D_{\max}^{\varepsilon'}(\E(\nu)\|\F(\nu))$, so that
		\begin{equation}
			\tau \leq 2^{D_{\max}^{\varepsilon'}(\E(\nu)\|\F(\nu))} \F(\nu)\,.
		\end{equation}
		Combining the last two inequalities leads to
		\begin{equation}
			\tau \leq 2^{D_{\max}^{\varepsilon'}(\E(\nu)\|\F(\nu)) + D_{\max}^\varepsilon(\rho\|\sigma)} \F(\sigma)\,.
		\end{equation}	
		It remains to show that $P(\tau, \E(\rho)) \leq \varepsilon + \varepsilon'$ (where $P$ is the sine distance). 
		This follows from
		\begin{align}
			P(\tau, \E(\rho)) &\leq P(\tau, \E(\nu)) + P(\E(\nu), \E(\rho)) \\
			& \leq \varepsilon' + P(\nu,\rho) \\
			&\leq \varepsilon' + \varepsilon\,,
		\end{align}
		where we used the triangle inequality and the data-processing inequality for the sine distance (see, e.g., \cite{khatri_principles_2020}). 
	\end{proof}

	\subsection{Non-Asymptotic Bounds for the Smoothed Max-Relative Entropy}
	
	The asymptotic equipartition property for the smoothed max-relative entropy  states that \cite{tomamichel_fully_2009}
	\begin{equation}
		\lim_{n \to \infty} \frac{1}{n} D^{\varepsilon}_{\max}(\rho^{\otimes n}\|\sigma^{\otimes n}) = D(\rho\|\sigma)\quad \forall \varepsilon \in (0,1) \,.
	\end{equation}
	For our proof of \autoref{thm:seq_par}, we will require a bound on the speed of this convergence. Results in this direction have appeared before in the literature; however, none of the previous results turn out to be directly applicable for our purposes (see \autoref{remark:relation_to_tomamichel} and \autoref{remark:relation_to_second_order} below). To prove our result we follow the established path of using quantum \renyi divergences. One key ingredient is a bound on the distance between the Petz--\renyi divergence and relative entropy, where the following is a slight refinement of \cite[Lemma 6.3]{tomamichel_framework_2012}:
	\begin{lemma}\label{lem:renyi_alpha_continuity}
		Let $\rho, \sigma \in \DM$ be quantum states. For $\gamma \in (0,1]$, define
		\begin{equation}\label{eq:c_gamma_states}
			c_\gamma(\rho\|\sigma) \coloneqq \frac{1}{\gamma} \log( 2^{\gamma D_{1 + \gamma} (\rho\|\sigma)} + 2^{- \gamma D_{1 - \gamma}(\rho\|\sigma)} + 1)\,.
		\end{equation}
		Then, for all $\gamma \in (0, 1]$ and  $\delta \in (0, \frac{\gamma}{2}]$:
		\begin{align}\label{eq:renyi_upper_continuity}
			D_{1 + \delta}(\rho\|\sigma) &\leq D(\rho\|\sigma) +  \ln(2) \delta(c_\gamma(\rho\|\sigma))^2 \\
			&\leq D(\rho\|\sigma) + \delta (c_\gamma(\rho\|\sigma))^2  . \label{eq:renyi_upper_continuity_weaker}
		\end{align}
		Furthermore, if $D(\rho\|\sigma) < \infty$, then for all $\gamma \in (0, 1]$ and  $\delta \in(0, \frac{\gamma}{2}]$ 
		\begin{multline}
			D_{1 - \delta}(\rho\|\sigma) \geq D(\rho\|\sigma) \\- \ln(2) \delta (c_\gamma(\rho\|\sigma))^2 \cosh(\ln(2) \delta  c_\gamma(\rho\|\sigma)) \,.
		\end{multline}
		and for all $\delta \in (0, \frac{\log 3}{2c_\gamma(\rho\|\sigma)}]$:
		\begin{align}\label{eq:renyi_lower_continuity}
			D_{1 - \delta}(\rho\|\sigma) &\geq D(\rho\|\sigma) - \ln(2) \cosh(\log(3)/2) \delta (c_\gamma(\rho\|\sigma))^2 \\
			&\geq D(\rho\|\sigma) - \delta (c_\gamma(\rho\|\sigma))^2 \,.
		\end{align}
	\end{lemma}
	The proof of this lemma appears in Appendix~\ref{app:proof_renyi_alpha_continuity}.
	
	\begin{remark}
		The previously known bound \cite[Lemma 6.3]{tomamichel_framework_2012} states only an analogue of \eqref{eq:renyi_upper_continuity_weaker} and follows from \autoref{lem:renyi_alpha_continuity} after setting $\gamma = 1/2$. Note that in its analogue of \eqref{eq:renyi_upper_continuity_weaker}, \cite[Lemma 6.3]{tomamichel_framework_2012} also has a stronger constraint on  the range of $\delta$, which turns out not to be necessary.
	\end{remark}
	Using \autoref{lem:renyi_alpha_continuity} we can then establish a bound on the convergence speed in the asymptotic equipartition property:
	\begin{lemma}\label{lem:aep_finite_n}
		Let $\rho, \sigma \in \DM$ be quantum states, and for $\gamma \in (0,1]$, take $c_\gamma(\rho\|\sigma)$ as in \eqref{eq:c_gamma_states}.
		Then, for all $\varepsilon \in (0, 1)$ and  $n \in \naturals$:
		\begin{multline}
			\label{eq:aep_finite_n_lower_bound_full}
			\frac{1}{n} D^\varepsilon_{\max}(\rho^{\otimes n} \|\sigma^{\otimes n}) \geq D(\rho\|\sigma) \\ -\frac{c_\gamma(\rho\|\sigma)}{\sqrt{n}}\left[\frac{\ln(2) \log(3)}{2} \cosh(\log(3) \over 2) \right. \\+ \left. \frac{4}{\log(3)} \log(1 \over 1 - \varepsilon)\right],
		\end{multline}
		\begin{multline} \label{eq:aep_finite_n_upper_bound_full}
			\frac{1}{n} D^\varepsilon_{\max}(\rho^{\otimes n} \|\sigma^{\otimes n}) \leq D(\rho\|\sigma) \\+ \frac{c_\gamma(\rho\|\sigma)}{\sqrt{n}}  \left[\frac{\ln(2) \log(3)}{2} +  \frac{4}{\log(3)} \log(1 \over \varepsilon)\right] \\
			+ \frac{1}{n} \log(1 \over 1 - \varepsilon^2),
		\end{multline}
		which implies
		\begin{align}\label{eq:aep_finite_n_lower_bound_reduced}
			\frac{1}{n} D^\varepsilon_{\max}(\rho^{\otimes n} \|\sigma^{\otimes n}) &\geq D(\rho\|\sigma) - \frac{4 c_\gamma(\rho\|\sigma)}{\sqrt{n}}\log(2 \over 1 - \varepsilon),
			\\ 
			\frac{1}{n} D^\varepsilon_{\max}(\rho^{\otimes n} \|\sigma^{\otimes n}) &\leq D(\rho\|\sigma) + \frac{4 c_\gamma(\rho\|\sigma)}{\sqrt{n}}  \log(2 \over \varepsilon) \notag \\ & \qquad + \frac{1}{n} \log(1 \over 1 - \varepsilon^2)\,.\label{eq:aep_finite_n_upper_bound_reduced}
		\end{align}
		\noindent These bounds can be tightened by adding a condition on $n$ to be large enough. Define first the numerical constants
		\begin{align}
			K_1 &\coloneqq  2 \sqrt{2 \ln(2) \cosh(\log(3) \over 2)} \leq 2.72,\\
			K_2 &\coloneqq  2 \sqrt{2 \ln(2)} \leq 2.36   \,.
		\end{align}
		Then, if
		\begin{equation}\label{eq:aep_finite_n_lower_range_n}
			n \geq \log(1\over 1 - \varepsilon)\qty(\frac{8}{\log(3) K_1})^2,
		\end{equation}
		we have the stronger bound
		\begin{equation}\label{eq:aep_finite_n_lower_bound}
			\frac{1}{n} D^\varepsilon_{\max}(\rho^{\otimes n} \|\sigma^{\otimes n}) \geq D(\rho\|\sigma) - \frac{K_1}{\sqrt{n}} c_\gamma(\rho\|\sigma) \sqrt{\log(1\over 1 - \varepsilon)},
		\end{equation}
		and similarly, if
		\begin{equation}\label{eq:aep_finite_n_upper_range_n}
			n \geq \log(1\over \varepsilon)\qty(\frac{8}{\gamma c_\gamma(\rho\|\sigma) K_2})^2,
		\end{equation}
		it holds that
		\begin{multline}\label{eq:aep_finite_n_upper_bound}
			\frac{1}{n} D^\varepsilon_{\max}(\rho^{\otimes n} \|\sigma^{\otimes n}) \leq D(\rho\|\sigma) + \frac{K_ 2}{\sqrt{n}} c_\gamma(\rho\|\sigma) \sqrt{\log(1\over\varepsilon)} \\+ \frac{1}{n} \log(1 \over 1 - \varepsilon^2)\,.
		\end{multline}
		
		\noindent Remember that $\gamma c_\gamma(\rho\|\sigma) \geq \log(3)$, and hence \eqref{eq:aep_finite_n_upper_range_n} is also always satisfied if
		\begin{equation}
			n \geq \log(1\over \varepsilon)\qty(\frac{8}{\log(3) K_2})^2\,.
		\end{equation}
	\end{lemma}
	The proof of this lemma appears in \Cref{app:proof_ape_finite_n}.
	
	\begin{remark}\label{remark:relation_to_tomamichel}
		Our \autoref{lem:aep_finite_n} can be seen an extension of \cite[Theorem 6.4]{tomamichel_framework_2012}, where a similar bound to \eqref{eq:aep_finite_n_upper_bound} is shown (however with a worse constant and different smoothing convention). Note that \cite{tomamichel_framework_2012} uses the notation $S(\rho\|\sigma) \coloneqq - D(\rho\|\sigma)$ and $S_{\min}^{\varepsilon}(\rho\|\sigma) \coloneqq - D_{\max}^{\varepsilon}(\rho\|\sigma)$, and hence \cite[Theorem 6.4]{tomamichel_framework_2012}, while looking like a lower bound, actually is an upper bound on $D_{\max}^{\varepsilon}$. An equivalent of \eqref{eq:aep_finite_n_lower_bound} is shown in \cite{tomamichel_framework_2012} only for the smoothed conditional min-entropy and not for the (more general) smoothed max-relative entropy. 
	\end{remark}

	\begin{remark}\label{remark:relation_to_second_order}
		Another already existing bound for the AEP convergence is the so-called second-order expansion \cite{tomamichel_hierarchy_2013}, which gives a tight asymptotic characterization also of the second-order $\sqrt{n}$ term in the convergence to the relative entropy. There, the second-order term is shown to be proportional to the square root of the relative entropy variance
		\begin{equation}
			V(\rho\|\sigma) \coloneqq \Tr(\rho (\log(\rho) - \log(\sigma) - D(\rho\|\sigma))^2)\,.
		\end{equation}
		Later in the proof of \autoref{thm:seq_par}, we want to apply these convergence bounds to the case in which $\rho = \rho_n$ and $\sigma = \sigma_n$ are the states in an adaptive protocol, and we would like to obtain a bound on the convergence parameter in $n$. We will see that by using chain rules for the geometric relative \renyi entropy (or the max-relative entropy) we will be able to show that $c_\gamma(\rho_n\|\sigma_n) = \mathcal{O}(n)$, while we are not aware of any way of obtaining such a bound for $V(\rho_n\|\sigma_n)$, and hence cannot directly use second-order asymptotics.
	\end{remark}

	\begin{EXCLUDED}
		\begin{lemma}\label{lem:aep_finite_n}
			Let $\rho$ and $\sigma$ be quantum states, and let $\varepsilon \in (0,1)$. Let $c(\rho\|\sigma)$ be defined as in \eqref{eq:c_rho_sigma}. Then, for all $n$:
			\begin{align}
				\frac{1}{n} D_{\max}^{\varepsilon}(\rho^{\otimes n}\|\sigma^{\otimes n}) &\leq D(\rho\|\sigma) + \frac{4 \sqrt{2}}{\sqrt{n}} c(\rho\|\sigma) \log(1 \over \varepsilon^2)  + \frac{1}{n} \log(1 \over 1 - \varepsilon^2) , \\
				\frac{1}{n} D_{\max}^{\varepsilon}(\rho^{\otimes n}\|\sigma^{\otimes n}) &\geq D(\rho\|\sigma) - \qty(\frac{4 \sqrt{2}}{\sqrt{n}} c(\rho\|\sigma) + \frac{1}{n}) \log(4 \over 1 - \varepsilon^2) \,.
			\end{align}
		\end{lemma}
		
		\begin{proof}
			The upper bound is immediate from combining the upper bound in \eqref{eq:anshu_min_max} with \eqref{eq:n_shot_steins_uppper}.
			For the lower bound, we combine the lower bound in \eqref{eq:anshu_min_max} with \eqref{eq:n_shot_steins_lower}, to find for any $\delta \in (0, 1 - \varepsilon^2)$
			\begin{equation}\label{eq:int_aep_lemm_step1}
				\frac{1}{n} D_{\max}^{\varepsilon}(\rho^{\otimes n}\|\sigma^{\otimes n}) \geq D(\rho\|\sigma) - \frac{4 \sqrt{2}}{\sqrt{n}} c \log(1 \over 1 -  \varepsilon^2 - \delta) - \frac{1}{n} \log(4 (1 - \varepsilon^2) \over \delta^2) + {2 \over n},
			\end{equation}
			where we abbreviated $c(\rho\|\sigma)$ as $c$.
			By taking the derivative of the right hand side with respect to $\delta$, we find that the right-hand side is stationary for 
			\begin{equation}
				-\frac{4 \sqrt{2}}{\sqrt{n}} c {1 \over 1 -  \varepsilon^2 - \delta} + \frac{2}{n \delta} = 0
			\end{equation}
			which is equivalent to 
			\begin{equation}
				\delta = \frac{1 - \varepsilon^2 - \delta}{\sqrt{8n} c}
			\end{equation}
			and hence
			\begin{equation}
				\delta =  \frac{1 - \varepsilon^2}{1 + \sqrt{8n}c}\,.
			\end{equation}
			
			We further find the second derivative of the right-hand side to be
			\begin{equation}
				\log(e)\qty[- \frac{4 \sqrt{2} c}{\sqrt{n}} \frac{1}{(1 - \varepsilon^2 - \delta)^2} - \frac{2}{n \delta^2}]
			\end{equation}
			which is negative for all $\delta \in (0, 1 - \varepsilon^2)$ and hence the function is concave in that range and the stationary point we just found is a global maximum. For simplicity let us define $K \coloneqq \sqrt{8n} c$. To simplify the right-hand side of \eqref{eq:int_aep_lemm_step1} at the optimal value of $\delta$, let us look at the two $\delta$-dependent terms, multiplied by $n$:
			\begin{equation}
				2K \log(1 - \varepsilon^2 - \delta) +\log(\delta ^2 \over 4 (1 - \varepsilon^2)) = (2K + 1) \log(1 - \varepsilon^2) + 2K \log(K \over 1 + K) + 2\log(1\over2(1 + K)) 
			\end{equation}
			The last two terms we can write as
			\begin{align}
				2(1 + K)\qty[{K \over 1 + K}\log(K \over 1 + K) + {1 \over 1 + K} \log(1 \over 1 + K) - {1 \over 1 + K}] &= 2 (1 + K)\qty[- h\qty(1 \over 1 + K) - {1 \over 1 + K}] \notag \\
				&\geq -4 (1 + K),
			\end{align}
			where $h$ is the binary entropy function. Putting everything together, \eqref{eq:int_aep_lemm_step1} with the optimal value of $\delta$ then becomes
			\begin{align}
				\frac{1}{n} D_{\max}^{\varepsilon}(\rho^{\otimes n}\|\sigma^{\otimes n}) &\geq D(\rho\|\sigma) - \frac{1}{n}\qty[(2K + 1)\log(1\over 1 - \varepsilon^2) + 4(1 + K) - 2] \\
				&= D(\rho\|\sigma) - \frac{(2K + 1)}{n}\log(4\over 1 - \varepsilon^2),
			\end{align}
			which is the claimed statement.
		\end{proof}
	\end{EXCLUDED}
	
	\subsection{Proof of our Main Result (\autoref{thm:seq_par})}
	\begin{proof}
		We start by applying \autoref{lem:wang_renner_upper_bound} to $D_H^{\alpha_a}(\E(\rho_n)\|\F(\sigma_n))$:
		\begin{multline}\label{eq:thm_starting_point}
			\frac{1}{n} D_H^{\alpha_a}(\E(\rho_n)\|\F(\sigma_n)) \\\leq \frac{1}{n} \frac{1}{1 - \alpha_a}\big(D(\E(\rho_n)\|\F(\sigma_n)) + h(\alpha_a)\big) \,.
		\end{multline}
		Note that a classical version of this equation in the context of channel discrimination was previously obtained in \cite[Eq. (33)]{hayashi_discrimination_2009}.
		Note now that we can write
		\begin{align} 
			&D(\E(\rho_n)\|\F(\sigma_n)) \\
			&= D(\E(\rho_n)\|\F(\sigma_n)) - D(\rho_n\|\sigma_n) + D(\rho_n\|\sigma_n) \label{eq:amort-1} \\
			&= D(\E(\rho_n)\|\F(\sigma_n)) - D(\rho_n\|\sigma_n) \notag \\ & \qquad + D(\Lambda_n(\E(\rho_{n-1}))\|\Lambda_n(\F(\sigma_{n-1}))) \\
			&\leq D(\E(\rho_n)\|\F(\sigma_n)) - D(\rho_n\|\sigma_n) + D(\E(\rho_{n-1})\|\F(\sigma_{n-1})) \\
			&\leq \ldots \leq \sum_{k = 1}^n \Big[ D(\E(\rho_k)\|\F(\sigma_k)) - D(\rho_k\|\sigma_k) \Big],
			\label{eq:amort-last}
		\end{align}
		where we used the definition of $\rho_k$ and $\sigma_k$, the data-processing inequality, and the fact that $\rho_1 = \sigma_1$. Let us use the index $\ell$ for the step in the adaptive protocol where this amortized difference is the largest, i.e.
		\begin{equation}
			\ell \coloneqq  \argmax_{k \in \{ 1, \ldots, n\}} \Big[ D(\E(\rho_k)\|\F(\sigma_k)) - D(\rho_k\|\sigma_k) \Big]\,.
		\end{equation}
		Then,    
		\begin{equation}\label{eq:rel_entropy_amortization}
			\frac{1}{n} D(\E(\rho_n)\|\F(\sigma_n)) \leq D(\E(\rho_\ell)\|\F(\sigma_\ell)) - D(\rho_\ell\|\sigma_\ell)\,.
		\end{equation}
		We can convert this to smoothed max-relative entropies by using \autoref{lem:aep_finite_n}. We will proceed with the bound requiring a condition on $m$ (or $n$ as it is called in \autoref{lem:aep_finite_n}); the $m$-independent bound is achieved in complete analogy by just taking the alternative statements from \autoref{lem:aep_finite_n}. We get:
		\begin{multline}\label{eq:amortized_to_max}
			D(\E(\rho_\ell)\|\F(\sigma_\ell)) - D(\rho_\ell\|\sigma_\ell) \leq \\  \frac{1}{m}\Big( D_{\max}^{\varepsilon_1} \big(\E^{\otimes m}(\rho_\ell^{\otimes m})\| \F^{\otimes m}(\sigma_\ell^{\otimes m})\big) - D_{\max}^{{\varepsilon_2}} \big(\rho_\ell^{\otimes m}\| \sigma_\ell^{\otimes m}\big) \Big) \\ 
			+ \frac{1}{\sqrt{m}}\left[K_1 c_{\gamma_1}\big(\E(\rho_\ell)\|\F(\sigma_\ell)\big) \sqrt{\log(1 \over 1 - \varepsilon_1)} \right. \\ + \left. K_2 c_{\gamma_2}(\rho_\ell\|\sigma_\ell) \sqrt{\log(1 \over \varepsilon_2)}\right] + \frac{1}{m}\log(1 \over 1 - \varepsilon_2^2)
		\end{multline}
		where $\gamma_1, \gamma_2 \in (0,1]$ and $\varepsilon_1, \varepsilon_2 \in (0,1)$ are arbitrary. It will be very convenient to choose $1-\varepsilon_1 = \varepsilon_2 \eqqcolon \varepsilon$, which is almost optimal as $K_1 \approx K_2$ and $c_\gamma\big(\E(\rho_\ell)\|\F(\sigma_\ell)\big) \approx c_\gamma(\rho_\ell\|\sigma_\ell)$ for large $\ell$ (which is the regime we are most interested in). 
		Then, 
		\begin{multline}\label{eq:amortized_to_max_2}
			D(\E(\rho_\ell)\|\F(\sigma_\ell)) - D(\rho_\ell\|\sigma_\ell) \leq \\
			\frac{1}{m}\Big( D_{\max}^{1 - \varepsilon} \big(\E^{\otimes m}(\rho_\ell^{\otimes m})\| \F^{\otimes m}(\sigma_\ell^{\otimes m})\big) - D_{\max}^{{\varepsilon}} \big(\rho_\ell^{\otimes m}\| \sigma_\ell^{\otimes m}\big) \Big) \\ 
			+ \frac{c_\ell}{\sqrt{m}} \sqrt{\log(1 \over \varepsilon)} + \frac{1}{m}\log(1 \over 1 - \varepsilon^2),
		\end{multline}
		where
		\begin{equation}
			c_\ell \coloneqq \inf_{\gamma_1, \gamma_2 \in (0,1]}\qty [K_1 c_{\gamma_1}\big(\E(\rho_\ell)\|\F(\sigma_\ell)\big) + K_2 c_{\gamma_2}(\rho_\ell\|\sigma_\ell)],
		\end{equation}
		and the conditions on $m$ from \autoref{lem:aep_finite_n} will be satisfied if
		\begin{equation}
			m \geq \log(1\over \varepsilon)\qty(\frac{8}{\log(3) K_2})^2\,.
		\end{equation}
		
		\noindent Taking $\varepsilon \leq 1/2$, we can apply \autoref{lem:chain_rule_simple} to \eqref{eq:amortized_to_max_2}. We then get a state $\tilde{\nu} \in B^\circ_\varepsilon(\rho_\ell^{\otimes m})$ such that
		\begin{multline}\label{eq:amortized_to_max_parralel}
			D(\E(\rho_\ell)\|\F(\sigma_\ell)) - D(\rho_\ell\|\sigma_\ell) \\\leq \frac{1}{m} D_{\max}^{1-2\varepsilon} \big(\E^{\otimes m}(\tilde{\nu})\| \F^{\otimes m}(\tilde{\nu})\big) 
			+ \frac{c_\ell}{\sqrt{m}} \sqrt{\log(1 \over \varepsilon)} \\ + \frac{1}{m}\log(1 \over 1 - \varepsilon^2)\, .
		\end{multline}
		
		Note that $\tilde{\nu} = \tilde{\nu}_{R_a^m A^m} \in \DM[R_a^{\otimes m} \otimes A^{\otimes m}]$ is the first of the two possible choices for $\nu$ given in \autoref{thm:seq_par}. For the other one, let $\tilde{\nu}_{R' R_a^m A^m}$ be a purification of $\tilde{\nu}_{R_a^m A^m}$, and hence also a purification of $\tilde{\nu}_{A^m}$. As the smoothed max-divergence satisfies the data-processing inequality (see, e.g., \cite[Prop.~4.60]{khatri_principles_2020}) we have
		\begin{multline}
			D_{\max}^{1-2\varepsilon} \big(\E^{\otimes m}(\tilde{\nu}_{R_a^m A^m})\| \F^{\otimes m}(\tilde{\nu}_{R_a^m A^m})\big) \\\leq D_{\max}^{1-2\varepsilon} \big(\E^{\otimes m}(\tilde{\nu}_{R' R_a^m A^m})\| \F^{\otimes m}(\tilde{\nu}_{R' R_a^m A^m})\big)\,.
		\end{multline}
		Now, let $\nu$ be the canonical purification of $\tilde{\nu}_{A^m}$, as specified in \autoref{thm:seq_par}. Since all purifications are equivalent up to an isometry on the purifying system (see, e.g., \cite{khatri_principles_2020}) on which the channels $\E^{\otimes m}$ and $\F^{\otimes m}$ act as an identity, and since the smoothed max-divergence is invariant under isometries, we have 
		\begin{multline}
			D_{\max}^{1-2\varepsilon} \big(\E^{\otimes m}(\tilde{\nu}_{R' R_a^m A^m})\| \F^{\otimes m}(\tilde{\nu}_{R' R_a^m A^m})\big) \\= D_{\max}^{1-2\varepsilon} \big(\E^{\otimes m}(\nu)\| \F^{\otimes m}(\nu)\big) \,.
		\end{multline}
		We will proceed with $\nu$ as the chosen state, but note that everything works analogously by choosing~$\tilde{\nu}$.  
		We now further apply the upper bound from \autoref{lem:anshu_min_max} to find:
		\begin{multline}\label{eq:amortized_to_HT}
			D(\E(\rho_\ell)\|\F(\sigma_\ell)) - D(\rho_\ell\|\sigma_\ell) \\\leq \frac{1}{m} D_{H}^{1 - (1- 2 \varepsilon)^2} \big(\E^{\otimes m}(\nu)\| \F^{\otimes m}(\nu)\big)
			+ \frac{c_\ell}{\sqrt{m}} \sqrt{\log(1 \over \varepsilon)} \\+ \frac{1}{m}\qty[\log(1 \over 1 - \varepsilon^2) + \log(1 \over 1 - (1-2 \varepsilon)^2)].
		\end{multline}
		To get our desired expression we set $\alpha_p \coloneqq 1 - (1-2\varepsilon)^2$, or equivalently $\varepsilon =\frac{1}{2}(1 - \sqrt{1 - \alpha_p})$.
		Using the estimates 
		\begin{equation}
			\frac{1 - \sqrt{1 - \alpha_p}}{2} \geq \frac{\alpha_p}{4}
		\end{equation}
		and 
		\begin{equation}
			1 - \varepsilon^2 = \frac{1 + \sqrt{1 - \alpha_p}}{2} + \frac{\alpha_p}{4} \geq 1 - \frac{\alpha_p}{4} ,
		\end{equation}
		we arrive at the expression
		\begin{multline}\label{eq:amortized_to_final}
			D(\E(\rho_\ell)\|\F(\sigma_\ell)) - D(\rho_\ell\|\sigma_\ell) \leq \frac{1}{m} D_{H}^{\alpha_p} \big(\E^{\otimes m}(\nu)\| \F^{\otimes m}(\nu)\big)\\ 
			+ \frac{c_\ell}{\sqrt{m}} \sqrt{\log(4 \over \alpha_p)} + \frac{1}{m}\qty[\log(1 \over \alpha_p) - \log(1 - \frac{\alpha_p}{4})],
		\end{multline}
		which we can insert into \eqref{eq:rel_entropy_amortization} and then \eqref{eq:thm_starting_point} to get
		\begin{multline}
			\frac{1}{n} D_H^{\alpha_a}(\E(\rho_n)\|\F(\sigma_n)) \\\leq \frac{1}{1 - \alpha_a}\Bigg[
			\frac{1}{m} D_H^{\alpha_p}(\E^{\otimes m}(\nu)\|\F^{\otimes m}(\nu)) + \frac{c_\ell}{\sqrt{m}} \sqrt{\log(4 \over \alpha_p)} \\+  \frac{1}{m}\qty[\log(1 \over \alpha_p) - \log(1 - \frac{\alpha_p}{4})] + \frac{h(\alpha_a)}{n}\Bigg],
		\end{multline}
		which leads to the desired statement in \eqref{eq:thm_seq_par}. 
		
		\medskip
		For the bounds on $c_\ell$, we start with
		\begin{align}
			&c_\gamma(\E(\rho_\ell)\|\F(\sigma_\ell)) \notag\\
			& = \frac{1}{\gamma}\log( 2^{\gamma D_{1 + \gamma} (\E(\rho_\ell)\|\F(\sigma_\ell))} + 2^{- \gamma D_{1 - \gamma}(\E(\rho_\ell)\|\F(\sigma_\ell))} + 1)\notag \\ &
			\leq \frac{1}{\gamma}\log( 2^{\gamma \widehat{D}_{1 + \gamma} (\E(\rho_\ell)\|\F(\sigma_\ell))} + 2)\,,
		\end{align}
		where we used \eqref{eq:petz-to-geometric} for the inequality.
		Now, by repeated use of the chain rule for the geometric \renyi divergence \cite[Thm 3.4]{fang_geometric_2021} we get
		\begin{align}
			&\widehat{D}_{1 + \gamma}(\E(\rho_\ell)\|\F(\sigma_\ell)) \notag\\
			& \leq \widehat{D}_{1 + \gamma}(\E\|\F) + \widehat{D}_{1 + \gamma}(\rho_\ell\|\sigma_\ell)\\
			&= \widehat{D}_{1 + \gamma}(\E\|\F) + \widehat{D}_{1 + \gamma}(\Lambda_\ell(\E(\rho_{\ell-1}))\|\Lambda_\ell(\F(\sigma_{\ell-1}))) \\
			&\leq \widehat{D}_{1 + \gamma}(\E\|\F) + \widehat{D}_{1 + \gamma}(\E(\rho_{\ell-1})\|\F(\sigma_{\ell-1})) \\
			&\leq \ldots \leq \ell \widehat{D}_{1 + \gamma}(\E\|\F)\,.
		\end{align}
		Defining the corresponding channel quantity
		\begin{equation}
			\widehat{c}_\gamma(\E\|\F) \coloneqq \frac{1}{\gamma}\qty(2^{\gamma \widehat{D}_{1 + \gamma}(\E\|\F)} + 2) \, ,
		\end{equation}
		we find that 
		\begin{align}
			c_\gamma(\E(\rho_\ell)\|\F(\sigma_\ell)) &\leq \frac{1}{\gamma} \log(2^{\ell \widehat{D}_{1 + \gamma}(\E\|\F)} + 2) \\
			&\leq \frac{\ell}{\gamma} \log(2^{\widehat{D}_{1 + \gamma}(\E\|\F)} + 2) \\&= \ell\, \widehat{c}_\gamma(\E\|\F) \,.
		\end{align}
		Using the same argument, we find that also $\widehat{D}_{1 + \gamma}(\rho_\ell \|\sigma_\ell) \leq \ell \widehat{D}_{1 + \gamma}(\E\|\F)$ and hence also
		\begin{equation}
			c_\gamma(\rho_\ell\|\sigma_\ell) \leq \ell\, \widehat{c}_\gamma(\E\|\F)\,.
		\end{equation}
		Thus,
		\begin{align}
			c_\ell &\leq \ell (K_1 + K_2) \inf_{\gamma \in (0,1]} \widehat{c}_\gamma(\E\|\F) \\ &\leq n (K_1 + K_2) \inf_{\gamma \in (0,1]} \widehat{c}_\gamma(\E\|\F)\,,
		\end{align}
		concluding the proof.
	\end{proof}
	
	\subsubsection{Proof of \autoref{cor:seq_par}}
	\begin{proof}
		\autoref{cor:seq_par} follows from \eqref{eq:seq_par_simple} by making the following estimates
		\begin{align}
			h(\alpha_a) &\leq 1 \, ,\\
			K &\leq 1 \, ,\\
			c'_\ell &\leq \frac{8 n}{\log(3)} \widehat{c}_{1}(\E\|\F) \approx 5.05\, n\,  \widehat{c}_{1}(\E\|\F) \notag\\& \leq 6\, n\, \widehat{c}_{1}(\E\|\F) \, ,\\
			- \log(1 - {\alpha_p \over 4}) &\leq 1 \, .
		\end{align}
		Combining all of these we can bound the error term by
		\begin{multline}
			\frac{6n}{\sqrt{m}} \widehat{c}_{1}(\E\|\F) \log(8 \over \alpha_p) + \frac{1}{m} \log(2 \over \alpha_p) + \frac{1}{n} \\\leq \frac{7 n}{\sqrt{m}} \widehat{c}_{1}(\E\|\F) \log(8 \over \alpha_p) + \frac{1}{n} \, ,
		\end{multline}
		where we used that $\widehat{c}_\gamma(\E\|\F) \geq 1$ for all $\gamma \in (0,1]$.
		Note finally that $\widehat{D}_2(\E\|\F) = D_2(\E\|\F) \leq D_{\max}(\E\|\F)$, which leads to the desired expression of $C$ in \eqref{eq:cor_seq_par}.
	\end{proof}
	
	\section{An Example}
	\label{sec:example}
	In this section we provide an example that illustrates how adaptive and parallel strategies can differ in the finite-length regime, and how our result bounds this difference. Specifically, we construct an example where the number of channel uses required by a parallel strategy to match a specific adaptive strategy turns out to be  arbitrarily large. This demonstrates that the relation between adaptive and parallel strategies in the finite-length regime is in general complex, and one cannot expect a substantially simpler relationship between adaptive and parallel strategies 
	than what we obtain in \autoref{cor:seq_par} and \autoref{thm:seq_par}.
	
	Our example is inspired by the already mentioned example from \cite{harrow_adaptive_2010, salek_usefulness_2022} that shows a separation between the asymptotic error decay rate of adaptive and parallel strategies in the \emph{symmetric} setting. For a specific pair of channels ($\E$, $\F$), they construct an adaptive strategy with two channel uses that achieves perfect discrimination (i.e., type~I and type~II error are equal to zero), whereas for parallel strategies they upper bound the symmetric error exponent by a finite value (i.e., they upper bound the rate at which both errors can simultaneously go to zero); so in particular, perfect discrimination is never possible with parallel strategies. Looking at this example in the asymmetric setting, one finds that there exists an input state $\nu \in \DM[A]$ such that for any arbitrary type~I error $\alpha_p$ there exists an $m$ such that $D_H^{\alpha_p}((\E(\nu))^{\otimes m}\|(\F(\nu))^{\otimes m}) = \infty$; i.e., already with a parallel strategy with product inputs one can achieve zero type~II error with an arbitrary small type~I error if one only makes $m$ large enough. Hence, unlike the symmetric setting, in the asymmetric setting there is no asymptotic gap between adaptive and parallel strategies, but there is still a significant difference for finite $m$, as the adaptive strategy achieves $\alpha_a = 0, \beta_a = 0$ with only two channel uses, whereas the parallel strategy requires a large $m$ to achieve $\alpha_p \leq \varepsilon, \beta_p = 0$. As mentioned in \autoref{remark:harrow_example}, the two channels $\E$ and $\F$ used in this example have $D_{\max}(\E\|\F) = \infty$ and hence \autoref{thm:seq_par} does not immediately apply. 
	What we are going to do though, is use a slightly noisy version of this channel $\F$, which makes $D_{\max}(\E\|\F)$ finite.
	
	Define the channels $\E, \F : \DM[\complex^2 \otimes \complex^2] \to \DM[\complex^2]$ for $\kappa \in [0,1]$ as follows:
	\begin{align}
		\E(\rho \otimes \omega) &\coloneqq \ketbra{0} \bkbraket{0|\omega|0} + \ketbra{0} \bkbraket{0|\rho|0} \bkbraket{1|\omega|1} \notag \\
		&\qquad + \frac{\IdentityMatrix}{2} \bkbraket{11|\rho \otimes \omega|11}
		\label{eq:example-channel-1}
		\\
		\F(\rho \otimes \omega) &\coloneqq (1-\kappa) \bigg[\ketbra{+} \bkbraket{0|\omega|0}  \notag \\&\qquad + \ \ketbra{1} \bkbraket{+|\rho|+} \bkbraket{1|\omega|1} \notag \\ & \qquad + \frac{\IdentityMatrix}{2} \bkbraket{-1|\rho \otimes \omega\left|-1\right.}\bigg] + \kappa \frac{\IdentityMatrix}{2}.
		\label{eq:example-channel-2}
	\end{align}
	It is easy to see that 
	\begin{align}
		\E\Big(\E\big(\ketbra{00}\big) \otimes \ketbra{1}{1}\Big) &= \ketbra{0}, \\
		\F\Big(\F\big(\ketbra{00}\big) \otimes \ketbra{1}{1}\Big) &= (1 - \delta(\kappa)) \ketbra{1} + \delta(\kappa)\ketbra{0}\,.
	\end{align}
	where $\delta(\kappa) = (3 \kappa - \kappa^2)/4$. For $\kappa = 0$ (which corresponds to the original example in \cite{harrow_adaptive_2010}) we find $\delta(0) = 0$ and hence the above gives an adaptive strategy that makes the channels perfectly distinguishable with just two channel uses. The exact same strategy will become arbitrarily good if $\kappa$ is nonzero but small, specifically
	\begin{equation}
		\frac{1}{2} D^{0}_H(\E(\rho_2)\|\F(\sigma_2)) = - \frac{1}{2} \log(\delta(\kappa)) \label{eq:example_two_uses},
	\end{equation}   
	where the adaptive strategy has $\rho_1 = \sigma_1 = \ketbra{00}$ and $\rho_2 = \ketbra{01}$, $\sigma_2 = (1 - \kappa/2)\ketbra{+1} + \kappa/2 \ketbra{-1}$.
	We also find:
	\begin{align}
		&D(\E(\rho_1)\|\F(\sigma_1)) - D(\rho_1\|\sigma_1) = D(\E(\rho_1)\|\F(\rho_1)) \notag \\& \qquad = D(\ketbra{0}\|(1 - \kappa/2)\ketbra{+} + \kappa/2 \ketbra{-})\notag \\
		&\qquad = -\frac{1}{2}\qty[ \log(\kappa \over 2) + \log(1 - {\kappa \over 2})]\,, \label{eq:example-adaptive-diff-1} \\
		&D(\E(\rho_2)\|\F(\sigma_2)) - D(\rho_2\|\sigma_2)\notag \\
		&\qquad= \log(\delta(\kappa)) + \frac{1}{2}\qty[ \log(\kappa \over 2) + \log(1 - {\kappa \over 2})]\,, \label{eq:example-adaptive-diff-2}
	\end{align}
	where interestingly \eqref{eq:example-adaptive-diff-1} is always larger than \eqref{eq:example-adaptive-diff-2}, and so it is the first step that increases the distinguishability the most (when measured in terms of relative entropy). This also implies that this adaptive strategy is not asymptotically optimal, as 
	\begin{equation}
		\frac{1}{2} D(\E(\rho_2)\|\F(\sigma_2)) < D(\E(\rho_1)\|\F(\rho_1)) \, ,
	\end{equation}
	and hence it is asymptotically better to just use a parallel strategy with tensor copies of $\rho_1$ as an input state. 
	Nevertheless, we will see that the adaptive strategy is still far superior in a regime where the number of channel uses $m$ is not too large: It is well known (see, e.g., \cite[Prop.~4.66]{khatri_principles_2020}) that for all states $\rho, \sigma$
	\begin{equation}
		\lim_{\alpha_p \to 0} D_H^{\alpha_p}(\rho\|\sigma) = D_H^{0}(\rho\|\sigma) = -\log(\Tr(\rho^0 \sigma))\,.
	\end{equation}
	In this specific example, for all input states $\nu \in \DM[\HS_R \otimes \mathbb{C}^{2} \otimes \mathbb{C}^{2}]$, it can be shown that
	\begin{equation}
		D_H^0(\E(\nu)\|\F(\nu)) \leq 2\,.
	\end{equation}
	To see this, we start by observing that the second input system of our channels (previously called $\omega$) can be treated as classical, and hence the maximum is attained at either $\omega = \ketbra{0}{0}$ or $\omega = \ketbra{1}{1}$ (this follows from joint convexity). If we choose the former, the channel output does not depend on the remaining input state and one finds 
	\begin{multline}
		D_H^0(\E(\rho_{RA} \otimes \ketbra{0}{0})\|\F(\rho_{RA} \otimes \ketbra{0}{0})) \\= -\log\Tr \big(\ketbra{0}{0} (\ketbra{+} (1 - \kappa/2) + \ketbra{-}(\kappa/2))\big) = 1\,.
	\end{multline}
	For the second choice of $\omega$ one finds that $D_H^0(\E(\nu)\|\F(\nu)) = 0$ unless $\nu = \rho_R \otimes \ketbra{0}{0} \otimes \ketbra{1}{1}$ and then 
	\begin{multline}
		D_H^0(\E(\ketbra{0}{0} \otimes \ketbra{1}{1})\|\F(\ketbra{0}{0} \otimes \ketbra{1}{1})) \\= - \log\Tr\big(\ketbra{0}{0} (\IdentityMatrix/4 (1 + \kappa))\big) \leq 2\,.
	\end{multline}
	An analogous argument also immediately yields $D_H^0(\E^{\otimes m}(\nu)\|\F^{\otimes m}(\nu)) \leq 2 m$ for any (potentially entangled) $\nu$. Hence, for fixed $m$, the rate 
	\begin{equation}
		\frac{1}{m} D_H^{\alpha_p}(\E^{\otimes m}(\nu)\|\F^{\otimes m}(\nu))
	\end{equation}
	of any parallel strategy can be brought down to 2 by making the type~I error threshold $\alpha_p$ small enough, whereas the mentioned adaptive strategy achieves zero type~I error and a type~II error rate that becomes arbitrarily large as $\kappa \to 0$.

	To actually calculate the performance of a parallel strategy where the input $\rho_1$ is used $m$ times, we can use the second-order asymptotics of hypothesis testing relative entropy to relative entropy \cite{li_second-order_2014} (this only works here because our input state is a product state). Define $\rho = \E(\rho_1)$, $\sigma = \F(\rho_1)$. Then \cite[Thm. 5]{li_second-order_2014} implies
	\begin{multline}\label{eq:example_second_order_lower}
		\frac{1}{m} D^{\alpha_p}_H(\rho^{\otimes m}\|\sigma^{\otimes m}) \\\geq D(\rho\|\sigma)  + \sqrt{V \over m} \Phi^{-1}\qty( \alpha_p - \frac{1}{\sqrt{m}} \frac{C T^3}{\sqrt{V^3}})\,,
	\end{multline}
	\begin{multline}
		\frac{1}{m} D^{\alpha_p}_H(\rho^{\otimes m}\|\sigma^{\otimes m}) \leq D(\rho\|\sigma) \\
		+ \sqrt{V \over m} \Phi^{-1}\qty(\alpha_p + \frac{1}{\sqrt{m}} \qty(\frac{C T^3}{\sqrt{V^3}} + 2) ) \,, \label{eq:example_second_order_upper}
	\end{multline}
	where $\Phi^{-1}$ is the inverse of the cumulative distribution function of the standard normal distribution, $C \leq 0.4784$, and
	\begin{align}
		V &\equiv  V(\rho\|\sigma) \coloneqq \Tr[\rho(\log \rho - \log \sigma - D(\rho\|\sigma))^2],\\
		T^3 &\equiv  T^3(\rho\|\sigma) \\&\coloneqq \sum_{i,j} \lambda_i \mabs{\braket{x_i}{y_j}}^2 \mabs{\log(\lambda_i) - \log(\mu_j) - D(\rho\|\sigma)}^3
	\end{align}
	for spectral decompositions $\rho = \sum_i \lambda_i \ketbra{x_i}$ and $\sigma = \sum_j \mu_j \ketbra{y_j}$.
	
	Note that, instead of comparing the type~II error decay rate of a parallel strategy with $m$ channel uses to the corresponding rate of an adaptive strategy with two channel uses, it might be more intuitive to compare the parallel strategy to a setup where the adaptive strategy with two channel uses is repeated $m/2$ times in parallel (so that the adaptive and parallel strategies have the same number of channel uses, and comparing rates is the same as comparing type~II errors). For this example, since the type~I error of the adaptive strategy was chosen to be zero, this is actually equivalent, as
	\begin{multline}
		\frac{1}{2k} D^{0}_H((\E(\rho_2))^{\otimes k}\|(\F(\sigma_2))^{\otimes k}) \\= \frac{1}{2} D^{0}_H(\E(\rho_2)\|\F(\sigma_2)) = - \frac{1}{2} \log(\delta(\kappa)) .
	\end{multline}

	Our theorem (\autoref{thm:seq_par}) gives an upper bound on the extent to which all finite-length parallel strategies can have worse type~II errors compared to the adaptive strategy, or equivalently how large~$m$ has to be chosen to achieve similar performances. For this example specifically, due to \eqref{eq:example-adaptive-diff-1} and \eqref{eq:example-adaptive-diff-2} we can choose $\ell = 1$ in \autoref{thm:seq_par}, and hence also the parallel input state $\nu$ in \autoref{thm:seq_par} will just be $\nu = \rho_1^{\otimes m}$.
	
	\begin{figure*}[!ht]
		\centering
		\includegraphics[width=\linewidth]{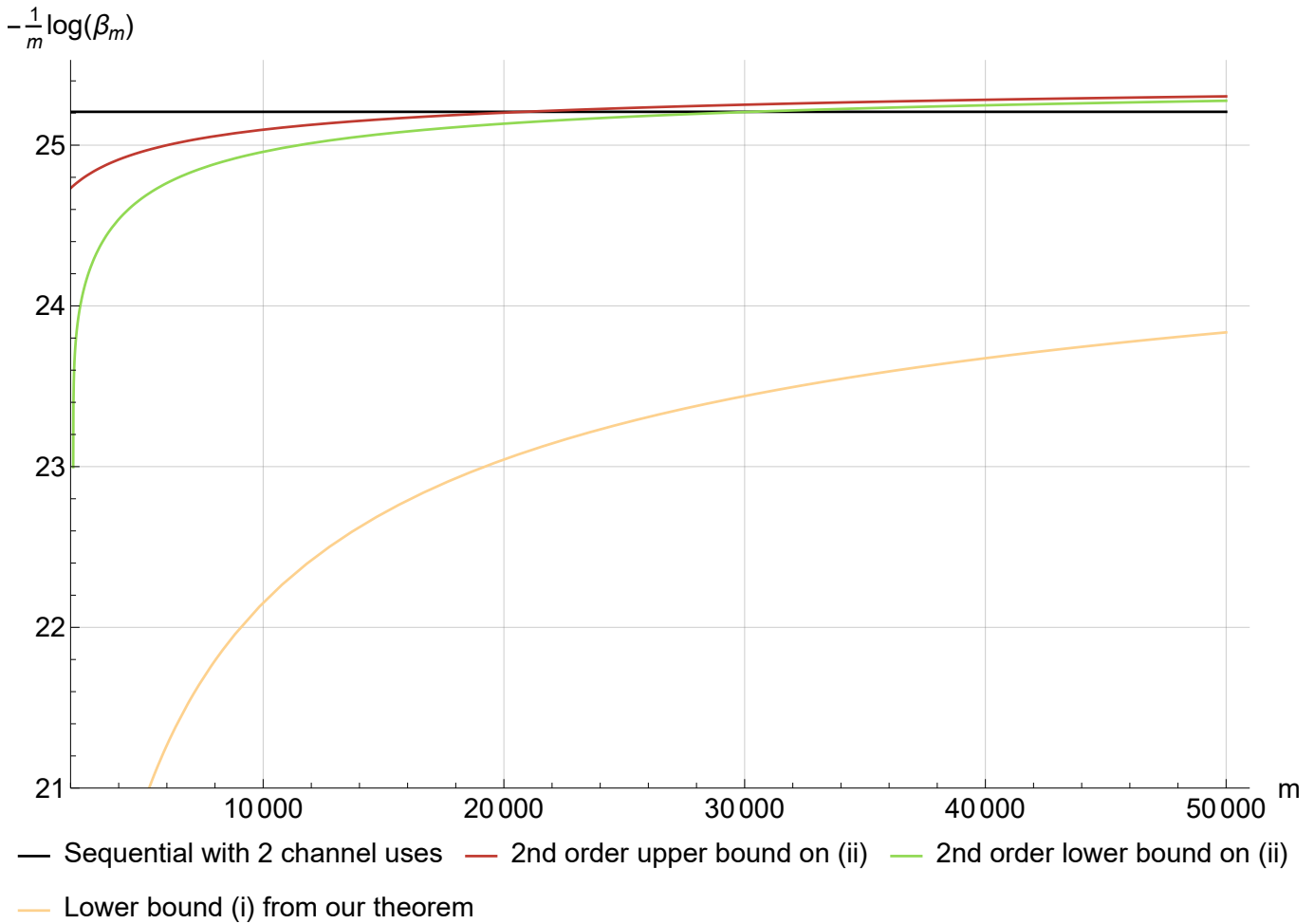}
		\caption{Illustration of the type~II error decay rate per channel use of a simple adaptive and parallel strategy for a specific pair of channels (see \eqref{eq:example-channel-1} and \eqref{eq:example-channel-2} for the definitions of the channels $\E$ and $\F$, respectively, where $\kappa = 2^{-50}$). We compare a fixed adaptive strategy with two channel uses (constant black line) to (i) our lower bound on the performance of a parallel strategy (yellow line) and (ii) the actual performance of a parallel strategy (red and green lines), which are plotted as functions of the number of parallel channel uses $m$.  The black line shows the value of \eqref{eq:example_two_uses}, i.e., the type~II error exponent for the given adaptive strategy with two channel uses and type~I error $\alpha_a = 0$. This can alternatively be thought of as the rate of repeating the two-step adaptive strategy $m/2$ times in parallel. 
			The yellow line shows the lower bound on the parallel strategy from our theorem (i.e., the right-hand side of \eqref{eq:thm_seq_par}), choosing $\alpha_p = 2^{-5}$. For this specific example we can calculate the parallel input state $\nu$ of our theorem, and while we cannot explicitly find the optimal POVM and corresponding type~II error (i.e., we cannot explicitly calculate the left-hand side of \eqref{eq:thm_seq_par}), we can bound it from above and below using the second-order asymptotics of the hypothesis testing relative entropy, which is shown in the red and green lines, corresponding to the values of \eqref{eq:example_second_order_upper} and \eqref{eq:example_second_order_lower}. We see that for small $m$ there is a gap between the adaptive and parallel strategies; i.e., the adaptive strategy offers an advantage. This advantage disappears once $m$ gets larger and in this specific example the chosen adaptive strategy even eventually gets surpassed by the parallel strategy, as the adaptive strategy turns out not to be asymptotically optimal.}
		\label{fig:example_convergence_plot}
	\end{figure*}    
	
	\autoref{fig:example_convergence_plot} depicts our lower bound (from \autoref{thm:seq_par}) on the performance of a parallel strategy for the given adaptive strategy, together with the actual performance of the parallel strategy choosing $\rho_1^{\otimes m}$ as the input state. For the figure we chose $\kappa = 2^{-50}$, $\alpha_a = 0$, and $\alpha_p = 2^{-5}$. The parameters are chosen to make the following features nicely visible simultaneously in one plot: $(i)$ the range of $m$ where the parallel strategy is worse, $(ii)$ the range of $m$ where it surpasses the adaptive strategy, and $(iii)$ our bound. For $c_\ell$ in \autoref{thm:seq_par} we used \eqref{eq:definition_cl} with a numerical optimization over $\gamma_1$. One finds that there is a range of values for $m$ for which the adaptive strategy is better, and the parallel strategy is lower bounded fairly tightly by our bound. As the given adaptive strategy is not asymptotically optimal it is eventually surpassed by the parallel strategy.
	
	\section{Discussion and Outlook}
	
	In this final section, we discuss several pathways through which one could hope to improve or extend our result, together with a discussion of obstacles we encountered on these pathways, and their relation to different open problems in quantum channel discrimination.
	
	First of all, one might hope to be able to remove the factor of $1 - \alpha_a$ appearing in \eqref{eq:thm_seq_par}, perhaps at the cost of an additional error term proportional to $\log(1 - \alpha_a)$. If this additional error term decays in $m$ and $n$ (say, as long as $\alpha_a$ is bounded away from one), this would prove the strong converse property for quantum channel discrimination.
	To see this, suppose we could show something along the lines of the following inequality:
	\begin{multline}
		\frac{1}{m} D_H^{\alpha_p}(\E^{\otimes m}(\nu)\|\F^{\otimes m}(\nu)) \overset{?}{\geq} \frac{1}{n} D_H^{\alpha_a}(\E(\rho_n)\|\F(\sigma_n)) \\- \frac{C_1 n}{\sqrt{m}} \log(8 \over \alpha_p) - \frac{C_2}{n}\log(1 \over 1 - \alpha_a) \, .
	\end{multline}
	Now, fixing any $\alpha_a \in (0,1)$ and taking limits (in this order) $m \to \infty, n \to \infty, \alpha_p \to 0$, 
	we would find that for an arbitrary $\alpha_a \in (0,1)$
	\begin{align}
		D^{\text{reg}}(\E\|\F) &= \lim_{\alpha_p \to 0} \lim_{n \to \infty} D^{\alpha_p}_H(\E^{\otimes n}\|\F^{\otimes n}) \\
		&\overset{?}{\ge} \limsup_{n\to\infty}\frac{1}{n}D_H^{\alpha_a}(\E(\rho_n)\|\F(\sigma_n)),
	\end{align}
	i.e., allowing a finite type~I error $\alpha_a\in(0,1)$ does not improve the best achievable type~II error rate. This is precisely the strong converse property. Establishing the strong converse property for quantum channel discrimination is an important and very interesting open problem, which is still unsolved despite serious efforts. Significant progress has been made in \cite{fawzi_defining_2021}, and another  recent attempt was made in \cite{fang_towards_2022}; see also \cite{berta_gap_2022}. 
	
	One might hope to obtain such a variant of our bound without the factor $1 - \alpha_a$, by, instead of transitioning from hypothesis testing entropy to relative entropy in \eqref{eq:thm_starting_point}, moving to an $\alpha$-\renyi relative entropy instead and subsequently employing the known relations between \renyi relative entropies and smoothed max-relative entropies. It turns out that this will at some point require bounding the difference between \renyi relative entropies of order $1 - \alpha$ and $1 + \alpha$. This is possible using \autoref{lem:renyi_alpha_continuity}, however at the cost of an error term $(c_\gamma(\rho_n\|\sigma_n))^2$, which will scale quadratically in~$n$. As this does not pick up any dependence in $m$, such an approach will not lead to anything useful in the asymptotic limit, and is hence unsuccessful. 
	\medskip
	
	One further interesting question deals with the relation between $\ell$ and $n$ in \autoref{thm:seq_par}. One might hope that an optimal adaptive strategy achieves the maximum possible distinguishability gain per channel use, i.e.
	\begin{equation}\label{eq:discussion_amortized}
		D^A(\E\|\F) = \sup_{\rho, \sigma \in \DM[R \otimes A]} D(\E(\rho)\|\F(\sigma)) - D(\rho\|\sigma),
	\end{equation}
	after some finite number of channel uses, or at least comes very close to it. If that was the case, one could bound $\ell$ by a constant and would remove the dependence on $n$ in \eqref{eq:cor_seq_par}. However, note that the supremum in \eqref{eq:discussion_amortized} goes over arbitrarily large reference systems $R$ and hence does not need to be achieved, and we are not aware of any bounds on the system size $R$ to get close to the optimum value. Alternatively, since $D^A(\E\|\F) = D^{\text{reg}}(\E\|\F)$ also a bound on the speed of convergence of the regularized channel divergence would do, which we are however also not aware of (Ref.~\cite{fawzi_defining_2021} proves such a bound for the sandwiched \renyi divergence, which unfortunately cannot be extended straightforwardly to relative entropy). Hence, for now we have to assume that $\ell$ can in general be $n$, and that to adequately simulate an adaptive strategy with $n$ channel uses, one might need to employ more than $n^2$ parallel channel uses to keep the error term in our bound small. 
	
	While our result becomes quite powerful in the asymmetric asymptotic setting, \autoref{thm:seq_par} can in principle be applied for any combination of type~I and II errors. Hence, it could be interesting to apply it to scenarios where the type~I error decays with $m$, say as $\alpha_p = 2^{-k m}$ for some constant $k$. In this case, one would want to apply \eqref{eq:thm_seq_par}, to have at least a chance of the error term being bounded; however -- depending on $k$ -- the corresponding condition \eqref{eq:thm_seq_par_bound_m} might not always be fulfilled. It would be interesting to see whether this condition on $m$ could perhaps be relaxed, to allow for a wider range of $k$ in this specific scenario. 
	
	As mentioned already in \autoref{remark:relation_to_second_order}, our bounds could be significantly tightened if we were able to employ second-order expansions, which however requires controlling the variance of the relative entropy $V(\rho_n\|\sigma_n)$ in $n$, which we are unable to do. This seems to be again related to the strong converse problem, as second-order expansions for the hypothesis testing relative entropy imply the strong converse property. Already for the parallel case, if one was able to show
	\begin{equation}
		V(\E^{\otimes n}(\nu)\|\F^{\otimes n}(\nu)) = \mathcal{O}(n)\, ,
	\end{equation}
	where $\nu \in \DM[R^{\otimes n} \otimes A^{\otimes n}]$ is the optimal joint input state, this would be a very significant step towards proving the strong converse for quantum channel discrimination. Conversely, examples where this scales faster than linearly in $n$ are quite likely to lend themselves to counterexamples for the strong converse property. We are not aware of any such examples: we leave this question open for further studies. 
	
	Finally, we believe that it should be possible to generalize our result to the infinite-dimensional setting with only minor modifications, which we also leave open to future investigations.
	
	\bigskip 
	
	\section*{Acknowledgements}
	Bjarne Bergh acknowledges support from the UK Engineering and Physical Sciences Research Council (EPSRC) under grant number EP/V52024X/1. Robert Salzmann acknowledges support from the Cambridge Commonwealth, European and International Trust. Mark~M.~Wilde acknowledges support from the National Science Foundation under grant no.~2315398. The authors thank the anonymous referees for various helpful comments and suggestions. For the purpose of open access, the authors have applied a Creative Commons Attribution (CC BY) licence to any Author Accepted Manuscript version arising from this submission.
	
	\bibliographystyle{IEEEtran}
	\bibliography{IEEEabrv, ieee}

	\appendix
	
	\section{Proofs of AEP Bounds}
	
	\subsection{Proof of \autoref{lem:renyi_alpha_continuity} (Petz--\renyi Continuity Bound in \texorpdfstring{$\alpha$}{α})}
	
	\label{app:proof_renyi_alpha_continuity}
	
	\begin{proof}[Proof of \autoref{lem:renyi_alpha_continuity}]
		Note first that for \eqref{eq:renyi_upper_continuity}, where we did not explicitly require $D(\rho\|\sigma) < \infty$, we can restrict to that case, as otherwise also $D_{1 + \delta}(\rho\|\sigma) = \infty$ and the statement is trivial. Hence, we can assume $\sigma$ to be invertible (otherwise restrict to the subspace where $\sigma$ is supported). Define the operator $X = \rho \otimes (\sigma^{-1})^T$, and the canonical purification of $\rho$ on $\HS \otimes \HS$:
		\begin{equation}
			\ket{\phi} = \sum_i \sqrt{\rho}\ket{i} \otimes \ket{i},
		\end{equation}
		where $\ket{i}$ is an orthonormal basis of $\HS$. Then, for all $\delta \in [-1, 1]$:
		\begin{equation}
			D_{1 + \delta}(\rho\|\sigma) = \frac{1}{\delta}\log(\bkbraket{\phi|X^\delta|\phi})
		\end{equation}
		and 
		\begin{equation}
			D(\rho\|\sigma) = \bkbraket{\phi|\log(X)|\phi}\,.
		\end{equation}
		For all $t > 0$ and $\delta \in \mathbb{R}$, the first term in the expansion of $t^\delta$ around $\delta = 0$ is $\delta \ln(t)$. Hence let us write $t^\delta = 1 + \delta \ln(t) + r_\delta(t)$ where $r_\delta(t)\coloneqq t^\delta - \delta \ln(t) - 1$.
		Since $1 + x \leq e^x $ for all $x \in \reals$ we see that
		\begin{align}
			r_\delta(t) &= t^\delta - \delta \ln(t) - 1 \\
			& \leq t^\delta + e^{-\delta \ln(t)} - 2 \\
			&= e^{\delta \ln(t)} + e^{- \delta \ln(t)} - 2 \\
			&= 2(\cosh(\delta \ln(t)) - 1) \eqqcolon s_\delta(t).
		\end{align}
		It is easy to see that $s_{-\delta}(t) = s_{\delta}(t)$ and $s_\delta(t) = s_{\gamma \delta}(t^{1/\gamma})$ for all $\gamma \in \reals$. Also, it is easy to verify that $s_\delta(t)$ is monotonically increasing in $t$ for $t \geq 1$ and concave in $t$ if $\delta \leq \frac{1}{2}$ and $t \geq 3$. For all $t \geq 0$, either $t$ or $1/t$ will be larger than or equal to one, and so we can always use the monotonicity to write
		\begin{equation}
			s_{\delta}(t) = s_{\delta}\qty(\frac{1}{t}) \leq s_{\delta}\qty(t + \frac{1}{t})\,.
		\end{equation}
		Using this, we get for all $t \geq 0$ and $\gamma \in (0,1]$
		\begin{align}
			s_{\delta}(t) &= s_{\delta/\gamma}(t^{\gamma}) \leq s_{\delta/\gamma}\qty(t^\gamma + t^{-\gamma})\\
			&\leq s_{\delta/\gamma}\qty(t^\gamma + t^{-\gamma} + 1)\,.
		\end{align}
		It is easy to see that the real function $x + 1/x$ has a global minimum for $x > 0$ at $x = 1$ and hence the argument at the right hand side is guaranteed to be larger than $3$. Hence we can use the concavity of $s_{\delta/\gamma}$ to write
		\begin{align}
			\bkbraket{\phi|s_\delta(X)|\phi} &\leq \bkbraket{\phi|s_{\delta/\gamma}(X^\gamma + X^{-\gamma} + \IdentityMatrix)|\phi} \\
			&\leq s_{\delta/\gamma}(\bkbraket{\phi|X^\gamma + X^{-\gamma} + \IdentityMatrix|\phi}) \\
			&= s_{\delta/\gamma}(2^{\gamma c_\gamma(\rho\|\sigma)}),
		\end{align}
		where the second inequality is a well-known property of concave functions (Jensen's inequality) and $c_\gamma$ is defined in~\eqref{eq:c_gamma_states}. Finally, we use Taylor's theorem with the Lagrange form of the remainder to bound
		\begin{align}
			s_\delta(t) &= s_0(t) + \dv{\delta}s_\delta(t) \big|_{\delta = 0} \delta + \frac{1}{2}\dv[2]{\delta} s_\delta(t) \big|_{\delta = \xi} \delta^2 \\
			&= \delta^2 (\ln(t))^2 \cosh(\xi \ln(t)) \\
			&\leq \delta^2 (\ln(t))^2 \cosh(\delta \ln(t)),
		\end{align}
		for all $t \geq 0$ and some $\xi\in(0,\delta)$, where we have used that $s_0(t) = \dv{\delta}s_\delta(t)|_{\delta=0}=0$. Hence,
		\begin{multline}
			\bkbraket{\phi|s_\delta(X)|\phi} \leq s_{\delta/\gamma}(2^{\gamma c_\gamma(\rho\|\sigma)}) \\ \leq (\delta \ln(2) c_\gamma(\rho\|\sigma))^2 \cosh(\delta \ln(2) c_\gamma(\rho\|\sigma))\,.
		\end{multline}
		To finally prove our claims, let us start with the case where $\delta > 0$. Then, 
		\begin{align}
			&D_{1 + \delta}(\rho\|\sigma) = \frac{1}{\delta}\log(\bkbraket{\phi|X^\delta|\phi}) \\
			&= \frac{1}{\delta}\log(1 + \delta \ln(2)\bkbraket{\phi|\log(X)|\phi} + \bkbraket{\phi|r_\delta(X)|\phi}) \\
			&= \frac{1}{\delta}\log(1 + \delta \ln(2)D(\rho\|\sigma) + \bkbraket{\phi|r_\delta(X)|\phi}) \\
			&\leq  \frac{1}{\delta}\log(1 + \delta \ln(2)D(\rho\|\sigma) + \bkbraket{\phi|s_\delta(X)|\phi})\\
			&=\frac{1}{\delta}\log(1 + \delta \ln(2)D(\rho\|\sigma)) \notag \\& \qquad \qquad+ \frac{1}{\delta}\log(1 + \frac{\bkbraket{\phi|s_\delta(X)|\phi}}{1 + \delta \ln(2)D(\rho\|\sigma)}) \\
			&\leq D(\rho\|\sigma) + \frac{1}{\delta}\log(1 + \bkbraket{\phi|s_\delta(X)|\phi}) \\
			& \leq D(\rho\|\sigma) \notag \\
			& \quad + \frac{1}{\delta}\log(1 + (\delta \ln(2) c_\gamma(\rho\|\sigma))^2 \cosh(\delta \ln(2) c_\gamma(\rho\|\sigma))) \\
			& \leq D(\rho\|\sigma) + \delta \ln(2) (c_\gamma(\rho\|\sigma))^2,
		\end{align}
		where the third-to-last inequality uses $\log(1 + x) \leq \frac{x}{\ln(2)}$ and (for the second term) the fact that $\delta$ and $D(\rho\|\sigma)$ are non-negative. The final inequality follows from the fact that 
		$k^2 - \ln(1 + k^2 \cosh(k))$ is monotonically increasing in $k$ and hence positive.
		
		If $\delta < 0$, the issue arises that $1 + \delta \ln(2)D(\rho\|\sigma)$ no longer has to be greater than $1$ (it might in fact be negative) and so we have to apply a slightly different argument. Starting analogously, we get
		\begin{align}
			&D_{1 + \delta}(\rho\|\sigma) \notag\\
			& = \frac{1}{\delta}\log(1 + \delta \ln(2)D(\rho\|\sigma) + \bkbraket{\phi|r_\delta(X)|\phi}) \\
			&\geq  \frac{1}{\delta}\log(1 + \delta \ln(2)D(\rho\|\sigma) + \bkbraket{\phi|s_\delta(X)|\phi})\\
			&\geq D(\rho\|\sigma) + \frac{1}{\delta \ln(2)}\bkbraket{\phi|s_\delta(X)|\phi} \\
			& \geq D(\rho\|\sigma)  + \delta \ln(2) (c_\gamma(\rho\|\sigma))^2 \cosh(\delta \ln(2) c_\gamma(\rho\|\sigma)),
		\end{align}
		where we used $\log(1 + x) \leq \frac{x}{\ln(2)}$ in the second inequality. Finally, if we assume that $\abs{\delta} \leq \frac{\log 3}{2c_\gamma(\rho\|\sigma)}$, then $\ln(2) \cosh(\ln(3)/2) < 1$ and so we get
		\begin{equation}
			D_{1 + \delta}(\rho\|\sigma)  \geq D(\rho\|\sigma)  + \delta (c_\gamma(\rho\|\sigma))^2\,.
		\end{equation}
		Note that since $c_\gamma(\rho\|\sigma)\geq \frac{\log(3)}{\gamma}$ the condition $\abs{\delta} \leq \gamma/2$ is automatically fulfilled.  
	\end{proof}

	\subsection{Proof of \autoref{lem:aep_finite_n} (AEP Convergence Bound)}\label{app:proof_ape_finite_n}
	
	\begin{proof}[Proof of \autoref{lem:aep_finite_n}]
		We start with the proof of \eqref{eq:aep_finite_n_lower_bound}. We use \cite[Prop. 4]{wang_resource_states_2019}, which states that for all $\alpha \in [0, 1)$ and all $\varepsilon \in [0,1)$
		\begin{equation}
			D^\varepsilon_{\max}(\rho\|\sigma) \geq D_\alpha(\rho\|\sigma) + \frac{2}{\alpha - 1}\log(1 \over 1 - \varepsilon)\,.
		\end{equation}
		Note that the statement in \cite{wang_resource_states_2019} is given for smoothing in trace-distance, but since the trace distance is always less than the sine distance, for fixed $\varepsilon$ the trace-distance-smoothed max-divergence is smoothed over a larger ball, and hence also always smaller, which is why the statement for the purified-distance-smoothed max-divergence is implied. 
		We apply this to $\rho^{\otimes n}$ and $\sigma^{\otimes n}$ and use the additivity of the Petz-\renyi relative entropy to get:
		\begin{equation}
			\frac{1}{n} D^\varepsilon_{\max}(\rho^{\otimes n}\|\sigma^{\otimes n}) \geq D_\alpha(\rho\|\sigma) + \frac{1}{n} \frac{2}{\alpha - 1}\log(1 \over 1 - \varepsilon).
		\end{equation}
		Now, we combine this with \eqref{eq:renyi_lower_continuity} of \autoref{lem:renyi_alpha_continuity} to get
		\begin{multline}
			\frac{1}{n} D^\varepsilon_{\max}(\rho^{\otimes n}\|\sigma^{\otimes n}) \geq D(\rho\|\sigma) \\+ \ln(2) \cosh(\log(3)\over 2) c_\gamma(\rho\|\sigma)^2 (\alpha - 1) \\+ \frac{1}{n} \frac{2}{\alpha - 1}\log(1 \over 1 - \varepsilon),
		\end{multline}
		together with the condition $0 \leq 1 - \alpha \leq \frac{\log 3}{2c_\gamma(\rho\|\sigma)}$.
		We are now free to choose $\alpha$ to optimize the right-hand side. It is easy to see that the right-hand side will be maximal if both terms are equal, which is achieved for
		\begin{multline}
			1 - \alpha = \sqrt{\frac{2  \log(1 \over 1 - \varepsilon)}{n \ln(2) \cosh(\log(3)\over 2) c_\gamma(\rho\|\sigma)^2}} \\= \frac{4}{K_1 c_\gamma(\rho\|\sigma)} \sqrt{\frac{\log(1 \over 1 - \varepsilon)}{n}}.
		\end{multline}
		Hence we get
		\begin{equation}
			\frac{1}{n} D^\varepsilon_{\max}(\rho^{\otimes n}\|\sigma^{\otimes n}) \geq D(\rho\|\sigma) - \frac{K_1}{\sqrt{n}} c_\gamma(\rho\|\sigma) \sqrt{\log(1\over 1 - \varepsilon)},
		\end{equation}
		together with the condition 
		\begin{equation}
			n \geq \log(1\over 1 - \varepsilon) \qty(4 \over K_1 \log(3))^2\,.
		\end{equation}
		
		In order to get an expression without a condition on $n$ we have to choose a different $\alpha$. Choosing 
		\begin{equation}
			1 - \alpha = \frac{\log(3)}{2 c_{\gamma}(\rho\|\sigma) \sqrt{n}} 
		\end{equation}
		satisfies the condition on $1 - \alpha$ for all $n$ and gives:
		\begin{multline}
			\frac{1}{n} D^\varepsilon_{\max}(\rho^{\otimes n} \|\sigma^{\otimes n}) \geq D(\rho\|\sigma) \\- \frac{c_\gamma(\rho\|\sigma)}{\sqrt{n}}\left[\frac{\ln(2) \log(3)}{2} \cosh(\log(3) \over 2) \right. \\+ \left. \frac{4}{\log(3)} \log(1 \over 1 - \varepsilon)\right]\,.
		\end{multline}
		
		The simplification \eqref{eq:aep_finite_n_lower_bound_reduced} follows from $\log(3) > 1$ and 
		\begin{equation}
			\frac{\ln(2) \log(3)}{2} \cosh(\log(3) \over 2) < 1\,.
		\end{equation}
		
		The second equation \eqref{eq:aep_finite_n_upper_bound} is proved very analogously. We start again with a relation to a \renyi relative entropy for $\alpha>1$ \cite[Prop. 4.61]{khatri_principles_2020}:
		\begin{equation}
			D^\varepsilon_{\max}(\rho\|\sigma) \leq D_\alpha(\rho\|\sigma) + \frac{2}{\alpha - 1}\log(1 \over \varepsilon) + \log(1 \over 1 - \varepsilon^2)\,.
		\end{equation}
		Note that \cite[Prop. 4.61]{khatri_principles_2020} uses the sandwiched \renyi divergence, which however is always less than the Petz, and hence the above statement is implied. We again apply this to $n$ tensor powers of $\rho$ and $\sigma$ and combine it with~\eqref{eq:renyi_upper_continuity} from \autoref{lem:renyi_alpha_continuity} to obtain:
		\begin{multline}
			\frac{1}{n} D^\varepsilon_{\max}(\rho^{\otimes n}\|\sigma^{\otimes n}) \leq D(\rho\|\sigma) + \ln(2) c_\gamma(\rho\|\sigma)^2 (\alpha - 1) \\+ \frac{2}{n(\alpha - 1)}\log(1 \over \varepsilon) + \frac{1}{n} \log(1 \over 1 - \varepsilon^2),
		\end{multline}
		together with the condition 
		\begin{equation}
			\alpha - 1 \leq \frac{\gamma}{2}\,.
		\end{equation}
		The right hand side is minimized again for
		\begin{equation}
			\alpha - 1 =  \sqrt{\frac{2 \ln(2) \log(1 \over \varepsilon)}{n c_\gamma(\rho\|\sigma)^2}} = \frac{4}{K_2 c_\gamma(\rho\|\sigma)} \sqrt{\frac{\log(1 \over \varepsilon)}{n}},
		\end{equation}
		which gives
		\begin{multline}
			\frac{1}{n} D^\varepsilon_{\max}(\rho^{\otimes n}\|\sigma^{\otimes n}) \leq D(\rho\|\sigma) \\+\frac{K_2}{\sqrt{n}} c_\gamma(\rho\|\sigma) \sqrt{\log(1 \over \varepsilon)} + \frac{1}{n} \log(1 \over 1 - \varepsilon^2),
		\end{multline}
		as well as the condition
		\begin{equation}
			n \geq  \log(1\over \varepsilon)\qty(\frac{8}{\gamma c_\gamma(\rho\|\sigma) K_2})^2\,.
		\end{equation}
		Alternatively, choosing
		\begin{equation}
			\alpha - 1 = \frac{\log(3)}{2 c_\gamma(\rho\|\sigma) \sqrt{n}}
		\end{equation}
		will again give \eqref{eq:aep_finite_n_upper_bound_full}, and the simplified form \eqref{eq:aep_finite_n_upper_bound_reduced} follows by the same argument that $\log(3) > 1$ and 
		\begin{equation}
			\frac{\ln(2) \log(3)}{2} < 1\,.
		\end{equation}
		This concludes the proof.
	\end{proof}	
	
	\begin{IEEEbiographynophoto}{Bjarne Bergh} received a bachelors and masters degree in Physics at the University of Heidelberg, Germany in 2018 and 2020 respectively, and a masters degree in Applied Mathematics at the University of Cambridge, United Kingdom in 2019. He is currently a PhD student at the University of Cambridge, United Kingdom, working on problems in Quantum Information Theory.
	\end{IEEEbiographynophoto}
	
	\begin{IEEEbiographynophoto}{Nilanjana Datta}
		received the Ph.D. degree in mathematical physics from ETH
		Zürich, Switzerland, in 1996. From 1997 to 2000, she was a Post-Doctoral
		Researcher with the Dublin Institute of Advanced Studies, C.N.R.S. Marseille,
		and EPFL, Lausanne. In 2001, she joined the University of Cambridge, U.K.,
		as a Lecturer in mathematics with the Pembroke College, and as a member of
		the Statistical Laboratory, Centre for Mathematical Sciences. She is currently
		a Professor of quantum information theory with the Department of Applied
		Mathematics and Theoretical Physics, University of Cambridge, and a fellow
		of the Pembroke College. Her research interests include quantum information
		theory and mathematical physics.
	\end{IEEEbiographynophoto}
	
	\begin{IEEEbiographynophoto}{Mark M. Wilde}
		received the Ph.D. degree in electrical engineering from the
		University of Southern California, Los Angeles, California. He is an Associate
		Professor of Electrical and Computer Engineering at Cornell University. He
		is an IEEE Fellow, he is a recipient of the National Science Foundation
		Career Development Award, he is a co-recipient of the 2018 AHP-Birkhauser
		Prize, awarded to “the most remarkable contribution” published in the journal
		Annales Henri Poincare, and he is an Outstanding Referee of the American
		Physical Society. His current research interests are in quantum Shannon theory,
		quantum computation, quantum optical communication, quantum computational complexity theory, and quantum error correction.
	\end{IEEEbiographynophoto}
	
	\begin{IEEEbiographynophoto}{Robert Salzmann}
		received the Ph.D. degree in mathematics from the University of Cambridge, United Kingdom in 2023. He is currently a post-doctoral researcher at the École Normale Supérieure de Lyon working on problems in quantum information theory and mathematical physics.
	\end{IEEEbiographynophoto}
	
\end{document}